\documentclass[12pt,titlepage,pdflatex]{article}
\usepackage{setspace} 
\usepackage{natbib}
\usepackage[english]{babel}
\usepackage[margin=1in]{geometry}
\usepackage[OT1]{fontenc}
\usepackage{graphicx}
\usepackage{authblk}
\usepackage{algorithmic}
\usepackage{amssymb}
\usepackage{rotating}
\usepackage{amsmath, amsthm, amsfonts}

\usepackage{bm}
\allowdisplaybreaks

\newtheorem{thm}{Theorem}[]

\newtheorem{prop}{Proposition}[]

\providecommand{\keywords}[1]{\textbf{\textit{Key words---}} #1}

\begin{document}
\title{Sparse Clustering of Functional Data}
\date{\vspace{-5ex}}

\author[1]{Davide Floriello}
\author[2]{Valeria Vitelli}
\affil[1]{Computer Science and Software Engineering Department, University of Canterbury, New Zealand}
\affil[2]{Oslo Center for Biostatistics and Epidemiology, Department of Biostatistics, University of Oslo, Norway}

\doublespacing
\maketitle

\renewcommand{\abstractname}{Author's Footnote}
\begin{abstract}
Davide Floriello is PhD student, Computer Science and Software Engineering Department, University of Canterbury, Private Bag 4800, Christchurch, New Zealand (e-mail: davide. floriello@pg.canterbury.ac.nz); Valeria Vitelli is Post Doctoral Fellow, Oslo Center for Biostatistics and Epidemiology, Department of Biostatistics, University of Oslo, P.O.Box 1122 Blindern, NO-0317 Oslo, Norway (e-mail: valeria.vitelli@medisin.uio.no).
\end{abstract}

\renewcommand{\abstractname}{Abstract}
\begin{abstract}
We consider the problem of clustering functional data while jointly selecting the most relevant features for classification. This problem has never been tackled before in the functional data context, and it requires a proper definition of the concept of sparsity for functional data. Functional sparse clustering is here analytically defined as a variational problem with a hard thresholding constraint ensuring the sparsity of the solution. First, a unique solution to sparse clustering with hard thresholding in finite dimensions is proved to exist. Then, the infinite dimensional generalization is given and proved to have a unique solution. Both the multivariate and the functional version of sparse clustering with hard thresholding exhibit improvements on other standard and sparse clustering strategies on simulated data. A real functional data application is also shown.
\keywords{Sparse Clustering, Functional Data, Weighted Distance, Variational Problem}
\end{abstract}

\section{Introduction}

In a clustering problem, it is unlikely that the real underlying groups differ in all the features considered. In most situations, only a limited number of features is relevant to detect the clusters. This is a known fact, part of the folk knowledge on cluster analysis dealing with vector data in $\mathbb{R}^p,$ and it further complicates the analysis as much as the number of features $p$ is larger than the sample size $N$. We want to consider, here, methods that are able to cluster data while also selecting their most relevant features. Such clustering methods are called \textit{sparse}. Sparse clustering has, at least, three advantages: firstly, if only a small number of features separates the clusters, it might result in a more accurate identification of the groups when compared with standard clustering. Moreover, it helps the interpretation of the final grouping. Finally, it reduces the dimensionality of the problem.

Sparse clustering methods for multivariate data have already been proposed. When dimensional reduction is the major focus, a common approach to non-parametric classification is based on Principal Component Analysis \citep{GC2002,Liu2003}. However, the use of PCA does not necessary lead to a sparse solution, often not even efficient, since principal components are usually linear combinations of all the features considered, i.e. very few loadings are zero. Moreover, there is no guarantee that the reduced space identified via PCA contains the signal that one is interested in detecting via clustering (see the study perfomed in \cite{Chang83}). Indeed sparse PCA clustering methods for vector data have been recently proposed, see for instance \cite{Luss.et.al.10} and references therein.

Model-based approaches to clustering have also been extensively studied in recent years, and they generally assume data to be generated from a mixture of Gaussian distributions with $K$ components \citep{Fraley-Raftery-2002,McLachlan-Peel-2000}. In the past decade it has been realized that the performance of model-based clustering procedures can be degraded if irrelevant or noisy variables are present. Hence, many efforts have been made to develop variable selection for model-based clustering: the proposal by \cite{RD2006}, and its improvements proposed by \cite{MCMM2009a,MCMM2009b}, can be considered as the current state of the art on variable selection methods in model-based clustering. However, even if these approaches seem promising \citep{CMMRR2013}, none of them has been generalized to the functional data framework. 

A different approach, called \emph{regularization}, consists in enforcing the solution of the clustering problem to be sparse by a suitable penalization. For instance, the penalty can be added to the maximization of the log-likelihood \citep{PS2007,WZ2008,XPS2008}. \cite{fm} proposed a completely different strategy, named Clustering Objects on Subsets of Attributes (COSA), which can be seen as a weighted version of K-means clustering where different feature weights within each cluster are allowed. However, this proposal does not correspond to a truly sparse method, since all variables have non-zero weights for positive values of the regularization parameter $\lambda$.  \cite{twitt} introduced a framework for regularization-based feature selection in clustering, which makes use of a Lasso type penalty. This corresponds to a soft thresholding operator, since it results in sparsity just for small values of the tuning parameter $s$ defining the $L^1$ constraint. \cite{twitt} decline their proposal both in the K-means case, and for hierarchical clustering procedures, also proposing a nice strategy to tune the sparsity parameter $s$ on the basis of a GAP statistics.

We aim here at developing a method for sparse clustering of functional data. Except for proposals for clustering sparsely observed curves \citep{James2013}, the only contribution in this direction is \citet{Tarpey14}, where a weighted $L^2$ distance among functional data is proposed to improve the performance of various statistical methods for functional data, including clustering. Even though the optimality criterion is quite different, we follow the same direction, and we propose a generalization of sparse clustering to functional data which is based on the estimation of a suitable weighting function $w(\cdot)$, capable of detecting the portions of the curves domain which are the most relevant to the classification purposes. The generalization is then natural if we define clustering as a variational problem, having $w(\cdot)$ and the data grouping as solutions, and enforce sparsity via a proper penalty or constraint on $w(\cdot)$. Note that variable selection in model-based clustering requires strong modelling assumptions, which are often too restrictive in the functional context; this is the main reason why a regularization approach is preferred in this context. However, differently from \cite{twitt}, the constraint ensuring the sparsity of the solution is a hard thersholding operator, meaning that the solution $w(\cdot)$ has to be non-zero only on a portion of the domain. As a by product of this formulation, we also derive a new formulation for the sparse clustering problem in finite dimensions, with the hard thersholding operator. In this sense, our proposal in finite-dimensions is a limiting case of the lasso-based proposal by \cite{twitt}.

The rest of the paper is organized as follows. In the next section, we briefly review the sparse clustering framework proposed in \cite{twitt} for finite dimensional data, and we recall their main result with the aim of helping the reader in noticing the analogies and differences with our approach. We introduce our hard thresholding proposal to sparse clustering in the finite dimensional setting, and we state the existence and uniqueness of the solution to this problem. In Section \ref{sec:infinite_dim}, after defining the theoretical setting and focussing on the meaning of features selection in infinite dimensions, we state the variational problem defining sparse clustering for functional data, and prove the existence and uniqueness of its solution. In Section \ref{sec:simulations} we test our method on synthetic data, while also comparing it to other multivariate and functional clustering approaches. An application to the analysis of the Berkeley Growth Study data, a benchmark in functional data analysis, is described in Section \ref{sec:realdata}. We conclude the paper in Section \ref{sec:discussion} with a discussion, where we also give a perspective on future developments. All simulations and analysis of real datasets are performed in \textsf{R} \citep{R}.

\section{Multivariate Sparse Clustering}\label{sec:finite_dim}

Let $\mathbb{X}\in \mathbb{R}^{N \times p}$ be the data matrix, and let us indicate with $\mathbf{X}_j$, $j=1,\ldots,p$ the $j$-th column of $\mathbb{X}$, including the values of the $j$-th feature for all data in the sample. Many clustering methods can be expressed as a maximization problem of the form
\begin{equation}\label{eq:twitt1}
 \max_{C_1,\ldots,C_K} \sum^p_{j=1} g_j(\mathbf{X}_j;C_1,\ldots,C_K),
\end{equation}
where $g_j(\mathbf{X}_j;C_1,\ldots,C_K)$ is some function that involves only the $j$-th feature of the data, and $C_1,\ldots,C_K$ is a partition of the $N$ data into $K$ disjoint sets. K-means clustering lies in this general framework, with $g_j(\cdot)$ being the between clusters sum of squares for feature $j$. \cite{twitt} define sparse clustering as the solution to the following maximization problem
\begin{align}\label{eq:twitt2}
& \max_{\mathbf{w}\in\mathbb{R}^p;C_1,\ldots,C_K} \sum^p_{j=1} w_j \cdot g_j(\mathbf{X}_j;C_1,\ldots,C_K), \\
& \text{subject to}\ \Vert \textbf{w} \Vert ^2_{\ell_2} \leq 1,\  \|  \textbf{w}  \|_{\ell_1}  \leq s,\ w_j \geq 0\ \text{for}\ j=1,\ldots,p. \nonumber
\end{align}
The weighting vector $\mathbf{w} \in \mathbb{R}^p $ has the role of feature selector, achieved by the lasso-type constraint on the $\ell_1$ norm. The other constraints are obvious: $\mathbf{w}$ must have at most unitary Euclidean norm, and it must be non-negative. The tuning parameter $s$ is specified by outer information.

The proposal in \cite{twitt} is to tackle the optimization problem (\ref{eq:twitt2}) via an iterative algorithm, that alternatively seeks the maximum in $\textbf{w}$ by holding $C_1,\ldots,C_K$ fixed, or finds the best partition $C_1,\ldots,C_K$ by holding the weighting vector $\textbf{w}$ fixed, until no significant variation in the weight vector \(\textbf{w}\) is observed. Convergence to a global optimum is not achieved in general, but we are guaranteed to increase the objective function at each iteration. When $\textbf{w}$ is fixed, the best partition is found by using a standard clustering procedure to a weighted version of the data. When the aim is optimizing (\ref{eq:twitt2}) with respect to $\textbf{w}$ by holding $C_1,\ldots,C_K$ fixed, the problem can be written as
\begin{align}\label{eq:twitt3}
& \max_{\mathbf{w}\in\mathbb{R}^p} \mathbf{w}^{\prime}\mathbf{a}, \\
& \text{subject to}\ \Vert \textbf{w} \Vert ^2_{\ell_2} \leq 1,\  \|  \textbf{w}  \|_{\ell_1}  \leq s,\ w_j \geq 0\ \text{for}\ j=1,\ldots,p, \nonumber
\end{align}
where $\mathbf{a}\in\mathbb{R}^p$ is such that its $j$-th component is $a_j = g_j(\mathbf{X}_j;C_1,\ldots,C_K)$. For $x \in \mathbb{R}$ and \(c\geq 0,\) define $S(x, c) = \text{sign}(x) (\vert x \vert - c)_+$ to be the soft thresholding operator. For ease of notation we allow for the operator \(S\) to be applied component-wise to vectors of \(\mathbb{R}^p.\) Then, the following result holds, directly from Karush-Kuhn-Tucker conditions (\cite{BoydVande})
\begin{prop}[\textbf{\cite{twitt}}]\label{prop:KKT}\  \\
 For $1 \leq s \leq \sqrt{p},$ a solution to problem (\ref{eq:twitt3}) is $$\textbf{w} = \frac{S(\textbf{a}_+, \Delta)}{\Vert S(\textbf{a}_+, \Delta) \Vert _{\ell^2}},$$ where $x_+$ denotes the positive part of $x\in \mathbb{R}$ and $\Delta = 0$ if that results in $\|  \textbf{w}  \|_{\ell_1}  \leq s$; otherwise, $\Delta > 0$ is chosen to yield $\|  \textbf{w}  \|_{\ell_1}  = s.$
\end{prop}

Instead of enforcing sparsity with a lasso-type $\ell_1$ constraint on the norm of $\mathbf{w}$, our proposal consists in forcing a certain number of components of $\mathbf{w}$ to be exactly zero, i.e. reformulate the same problem with respect to a counting measure. This would mean defining \emph{multivariate sparse clustering} as the solution to the following maximization problem
\begin{align}\label{eq:finitesparsepb}
& \max_{\mathbf{w}\in\mathbb{R}^p;C_1,\ldots,C_K} \sum^p_{j=1} w_j \cdot g_j(\mathbf{X}_j;C_1,\ldots,C_K), \\
& \text{subject to}\ \Vert \textbf{w} \Vert ^2_{\ell_2} \leq 1,\ w_j \geq 0\ \text{for}\ j=1,\ldots,p\ \text{and}\ \mu^\sharp(\left\lbrace j: w_j=0 \right\rbrace ) = m,\nonumber
\end{align}
where \(\mu^\sharp\) indicates the counting measure on \(\Omega=\{1,...,p\}.\) The weighting vector $\mathbf{w} \in \mathbb{R}^p$ (feature selector) is now forced to have exactly $m$ components equal to zero, where \(m\) is an integer such that \(0 \leq m < p.\) There is a relation between the parameter $s$ in (\ref{prop:KKT}) and the parameter $m$ we are introducing here, since for small values of $s$ some components will indeed be null. However, the solution of problem (\ref{eq:finitesparsepb}) is a hard-thresholding solution, which the soft-thresholding solution found by optimizing (\ref{eq:twitt2}) can recover only for some limiting values of $s$. It will be shown via simulation studies that this hard thresholding strategy reduces the misclassification error in a high-dimensional (non necessarily functional) setting.

The optimization problem (\ref{eq:finitesparsepb}) can be tackled by an iterative algorithm. Hence, we only need to consider the following constrained optimization problem
\begin{align}\label{pr:sparseKm}
& \max_{\mathbf{w} \in \mathbb{R}^p} \mathbf{w}^{\prime}\mathbf{b} \\
& \text{subject to}\ \Vert \mathbf{w} \Vert_{\ell^2} \leq 1,\ w_j \geq 0\ \text{for} \ j=1,\ldots,p\ \text{and}\ \mu^\sharp(\left\lbrace j: w_j=0 \right\rbrace ) = m,\nonumber
\end{align}
where \(\mathbf{b}=(b_1,...,b_p)^\prime\) is a vector of \(\mathbb{R}^p,\) and we assume that the components of \(\mathbf{b}\) are all greater than zero and with no ties, i.e. $b_i>0\ \text{and}\ b_i\not=b_j,\ \mbox{for} \ i,j=1,...,p,\ i\not = j$. The vector $\mathbf{b}$ has the role of describing the goodness of the partition with respect to each of the features. In order to represent the solution of problem (\ref{pr:sparseKm}), let \(b_{(i)}\) be the i-th ordered component of the vector \(\mathbf{b}\) and set $B=\{i\in \Omega: b_i > b_{(m)}\}$, letting \(B=\Omega\) if \(m=0,\) and \(B=\emptyset\) if \(m=p.\)

\begin{thm}\label{Teo_sparse_vectorial}\ \\
Problem (\ref{pr:sparseKm}) has a unique solution \(\mathbf{w^*} \in \mathbb{R}^p\) whose components are:
\[
w^*_i=\left \{ \begin{array}{ll}
              \frac{b_i}{(\sum_{j \in B} b_j^2)^{\frac{1}{2}}},& \text{if}\ i\in B,\\
              &\\
              0,&\text{otherwise}.
              \end{array}
    \right.
\]
\end{thm}

\proof\ Consider the function \(f(\mathbf{w})=\sum^p_{j = 1} w_jb_j\) defined on the domain $D=\{\mathbf{w} \in \mathbb{R}^p: \Vert \mathbf{w} \Vert_{\ell^2} \leq 1,\ w_j \geq 0\ \text{for} \ j=1,\ldots,p\ \text{and}\ \mu^\sharp(\left\lbrace j: w_j=0 \right\rbrace ) = m\}$. The function \(f\) is continuous on \(D\) and \(D\) is a compact subset of \(\mathbb{R}^p,\) being the finite union of closed and bounded sets. Therefore problem (\ref{pr:sparseKm}) has a solution. Let \(\mathbf{w}\) be a solution of problem (\ref{pr:sparseKm}) and set $A = \{j: w_j >0 \}$. The set \(A\) and the set \(B\) coincide. Indeed the two sets have the same number of elements. If they were not the same set, there would exist an \(i \in A\setminus B\) and a \(j \in B\setminus A\) such that, by switching the value \(w_i\) with the value \(w_j\), one could construct a \( \mathbf{\tilde {w}}\in \mathbb{R}^p\) (satisfying the constraints of problem (\ref{pr:sparseKm})) such that \(f(\mathbf{\tilde{w}})>f(\mathbf{w})\), against the assumption that \(\mathbf{w}\) is a solution of problem (\ref{pr:sparseKm}). In fact, set:
\begin{equation}\label{w tilde}
\left\{ \begin{array}{lcll}
\tilde{w_j} & = & w_i,\\
\tilde{w_i} & = & w_j,\\
\tilde{w_k} & = & w_k,& \ \text{for} \ k \in \Omega \setminus \{i,j\}.
\end{array}
\right.
\end{equation}
Then $\Vert \mathbf{\tilde{w}} \Vert_{\ell^2} \leq 1,\ \tilde{w}_j \geq 0\ \text{for} \ j=1,\ldots,p\ \text{and}\ \mu^\sharp(\left\lbrace j: \tilde{w}_j=0 \right\rbrace ) = m$, while
\begin{eqnarray}\label{w tilde ottima}
\sum^p_{j=1} w_jb_j
 &=&  \sum_{k \in \{1,...,p\}\setminus \{i,j\}} w_kb_k + w_jb_j +  w_ib_i \nonumber\\
 &<&  \sum_{k \in \{1,...,p\}\setminus \{i,j\}} w_kb_k + w_ib_j +  w_jb_i = \sum^p_{j=1} \tilde{w}_jb_j.
\end{eqnarray}
Thus $A$ and $B$ coincide. Therefore, if \(\mathbf{w}\) is a solution of (\ref{pr:sparseKm}),
\begin{eqnarray*}
\sum_{j=1}^p w_j b_j = \sum_{j \in A} w_j b_j = \sum_{j \in B} w_jb_j &\leq & (\sum_{j \in B} w^2_j)^{\frac{1}{2}}(\sum_{j \in B} b^2_j)^{\frac{1}{2}}\\
&\leq & (\sum_{j \in B} b^2_j)^{\frac{1}{2}}
\end{eqnarray*}
where the first inequality is Cauchy-Schwarz inequality and the second follows from the fact that \(\Vert \mathbf{w} \Vert_{\ell^2} \leq 1.\) Take \(\mathbf{w^*} \in \mathbb{R}^p\) with components
\[
w^*_i=\left \{ \begin{array}{ll}
              \frac{b_i}{(\sum_{j \in B} b_j^2)^{\frac{1}{2}}},& \text{if}\ i\in B,\\
              &\\
              0,&\text{otherwise}.
              \end{array}
    \right.
\]
Then \(\mathbf{w^*}\) satisfies the constraints of problem (\ref{pr:sparseKm}) and $\sum_{j=1}^p w^*_jb_j = (\sum_{j \in B} b^2_j)^{\frac{1}{2}}$. Hence, \(\mathbf{w^*}\) is a solution of problem (\ref{pr:sparseKm}).

Let \(\mathbf{\tilde{w}}\) be another solution of problem (\ref{pr:sparseKm}). Indicate with \(\mathbf{b}^*\) the vector of \(\mathbb{R}^p\) with components $b^*_i=b_i$ if $i \in B$, 0 otherwise, and with \(\theta\) the angle between \(\mathbf{\tilde{w}}\) and \(\mathbf{b}^*.\)
Then $\cos(\theta)= \frac{\sum_{j \in B} \tilde{w}_jb_j}{\Vert \mathbf{\tilde{w}} \Vert_{\ell^2} \Vert \mathbf{b}^* \Vert_{\ell^2}}
=\frac{1}{\Vert \mathbf{\tilde{w}} \Vert_{\ell^2}}$, where the second equality is true because $\sum_{j=1}^p \tilde{w}_jb_j = \sum_{j=1}^p w_jb_j = (\sum_{j \in B} b^2_j)^{\frac{1}{2}}=\Vert \mathbf{b}^* \Vert_{\ell^2},$ being \(\mathbf{\tilde{w}}\) a solution of (\ref{pr:sparseKm}). Since \(\Vert \mathbf{\tilde{w}} \Vert_{\ell^2} \leq 1,\) and $\cos(\theta)=1/\Vert \mathbf{\tilde{w}} \Vert_{\ell^2}$, then \(\Vert \mathbf{\tilde{w}} \Vert_{\ell^2}=1.\)
Hence \(\cos(\theta)=1,\) and
\[
\mathbf{\tilde{w}}= \frac{\mathbf{b}^*}{\Vert \mathbf{b}^* \Vert_{\ell^2}} = \mathbf{w},
\]
proving that the solution of problem (\ref{pr:sparseKm}) is unique.
\endproof

We remark that the uniqueness of the solution given by Theorem \ref{Teo_sparse_vectorial} depends on the hypothesis of having no ties in the vector $\mathbf{b}$, because ties make the sorting of $\mathbf{b}$ not unique. However, this assumption seems reasonable in practice, being $\mathbf{b}$ the vector describing cluster separation (e.g., the Between-Clusters Sum of Squares (BCSS) in case of K-means) with respect to each feature.

Having found the unique optimal weighting vector, given by Theorem \ref{Teo_sparse_vectorial}, we are guaranteed that also the solution to the finite dimensional sparse clustering problem (\ref{eq:finitesparsepb}) is unique, since the number of possible partitions of the data is finite. However, this theoretical guarantee does not solve the problem in practice, since finding the maximum over the set of all possible partitions is a combinatorial optimization problem, for which a complete enumeration search over all possible groupings of the data is computationally impractical in applications. A large number of heuristic search strategies have been proposed in the literature for ordinary weighted clustering based on criteria of the form (\ref{eq:twitt1}). Hence, the optimization problem (\ref{eq:finitesparsepb}) can be tackled by an iterative algorithm where the clustering step is specified via a proper search strategy algorithm. Details on the proposed procedure are given in the next subsection. The algorithm performance will be tested in Section \ref{sec:simulations} on synthetic data.

\subsection{A K-means implementation of Multivariate Sparse Clustering}

We are now going to set up an algorithm which is suitable for finding the solution to problems of the form (\ref{eq:finitesparsepb}), where the functions $g_j(\mathbf{X}_j;C_1,\ldots,C_K)$ correspond to the BCSS with respect to the $j$-th feature and in the $\ell_2$ distance
\begin{equation}\label{eq:BCSSfinite}
g_j(\mathbf{X}_j;C_1,\ldots,C_K) = \frac{1}{N}\sum_{i=1}^N\sum_{i^\prime=1}^N (x_{ij}-x_{i^\prime j})^2 - \sum_{k=1}^K \frac{1}{N_k} \sum_{i i^\prime \in C_k} (x_{ij}-x_{i^\prime j})^2,
\end{equation}
and where $N_k$ is the number of elements belonging to the $k$-th cluster. Hence, the criterion used in our algorithm assigns higher weights to features that can cause a higher increase in the BCSS, i.e. to features with respect to which clusters can be distinguished the most.

It is well-known from the literature on cluster analysis that the $K$-means algorithm \citep{po2,har,Hartigan-1978} is an optimal search strategy with respect to criterion (\ref{eq:finitesparsepb}), when (\ref{eq:BCSSfinite}) is chosen. Hence, we implement the following \emph{Multivariate Sparse $K$-means Clustering Algorithm}:
\begin{enumerate}
\item initialize $C_1,\ldots,C_K$ by running a simple multivariate $K$-means;
\item iterate until convergence:
\begin{itemize}
\item[(i)] holding $C_1,\ldots,C_K$ fixed, use Theorem \ref{Teo_sparse_vectorial} to obtain the optimal weighting vector $\mathbf{w}^*$, where the vector $\mathbf{b}=(b_1,\ldots,b_p)$ is such that $b_j=g_j(\mathbf{X}_j;C_1,\ldots,C_K)$ in (\ref{eq:BCSSfinite}) for each $j=1,\ldots,p$;
\item[(ii)] holding $\mathbf{w}=\mathbf{w}^*$ fixed, find the optimal partition $C_1,\ldots,C_K$ by optimizing criterion (\ref{eq:finitesparsepb}) with BCSS, i.e., run a $K$-means on the weighted data obtained by considering the $N \times N$ dissimilarity matrix whose $(i,i^\prime)$ element is $d_{i,i^\prime}=\sum_{j=1}^p w_j (x_{ij}-x_{i^\prime j})^2$;
\end{itemize}
\item stop when no more changes in the partition $C_1,\ldots,C_K$ are observed at two subsequent iterations.
\end{enumerate}
The optimal weighting vector is given by the $\mathbf{w}^*$ obtained at the last (i) step, while the final clusters are given by $C_1,\ldots,C_K$ in the last (ii) step.

A quite relevant issue remains open at this point: how to select the number $m$ of features that we expect to have associated weight $w_j$ null, for a given application? In simulation studies and real applications, $m$ can be tuned by using a permutation approach based on the computation of a GAP statistics, that has been proved to be quite successful in \cite{twitt} for tuning parameter $s$, and in \cite{twh2001} for choosing the number of clusters $K$.

\section{Functional Sparse Clustering}\label{sec:infinite_dim}

We now focus on the same problem when data lie in an infinite dimensional space. Let $(D, \mathcal{F}, \mu)$ be a measure space, with $\mu$ the Lebesgue measure, and $D \subset \mathbb{R}$ compact. Let \(f_1,...,f_N: D \rightarrow \mathbb{R}\) be a functional data set, with $f_i \in L^2(D)$ for all $i=1,\ldots,N,$ which is a reasonable assumption in most applications. We further assume that a proper functional form has been reconstructed for each datum in a preprocessing step of the analysis, thanks to some smoothing technique (see \cite{rs2} and references therein). We finally assume that, if a phase variability is present, data have already been aligned (see Section \ref{sec:discussion} for a discussion on this issue).

Likewise the finite dimensional case, many clustering methods for functional data can be written as a variational problem of the form
\begin{equation}\label{eq:ClustFunc}
\max_{C_1,\ldots,C_K} \int_D g(f_1(x),\ldots,f_N(x);C_1,\ldots,C_K) dx,
\end{equation}
where for ease of notation we write $dx$ instead of $d\mu(x)$. The function $g(\cdot)$ depends on the chosen clustering strategy; for instance, the standard K-means algorithm for functional data looks for a partition \(C_1,...,C_K\) of the data labels that maximizes the BCSS
\begin{align}\label{eq:KmeansFunc}
& g(f_1(x),\ldots,f_N(x);C_1,\ldots,C_K) = \\
& \frac{1}{2N} \sum^N_{i,j = 1} \left( f_i(x) - f_j(x) \right)^2 - \sum^K_{h=1} \frac{1}{\vert 2C_h \vert} \sum_{i,j \in C_h} \left( f_i(x) - f_j(x) \right)^2.\nonumber
\end{align}
We first need to give a meaning to the concept of feature selection for functional data, getting inspired by the finite-dimensional case. When data are elements of $\mathbb{R}^p,$ feature selection can be carried out by an optimal weighting vector $\bf{w} \in \mathbb{R}^p$. The weighting vector \(\bf{w}\) can be analytically computed according to Proposition \ref{prop:KKT} in the context of soft thresholding, and to Theorem \ref{Teo_sparse_vectorial} in the context of hard thresholding. The latter approach, which we suggest, consists in finding the subset $S$ of the set of $p$ features that discriminate the most between clusters. This idea can be adapted to functional data, provided we change the counting measure on the finite dimensional set of features with a proper measure \(\mu\) on the continuum features index set \(D\). In the functional setting the role of the feature selector can be thus played by a weighting function $w: D \rightarrow \mathbb{R}$. We define \emph{functional sparse clustering} as the solution to a variational constrained optimization problem, where the constraints on $w$ enforce the sparsity of the optimal solution
\begin{align}\label{eq:infinitesparsepb}
 \max_{w \in L^2(D); C_1, \ldots  C_K} & \int_{D} w(x) g(f_1(x),\ldots,f_N(x);C_1,\ldots,C_K) dx, \\
 \text{subject to:}\ \ & \Vert w(x) \Vert_{L^2(D)} \leq 1,\ \ w(x) \geq 0\  \mu-a.e.\ \text{and}\ \mu(\left\lbrace x \in D: w(x)=0\right\rbrace ) \geq m. \nonumber
\end{align}
The weighting function $w$ is non-negative almost everywhere on $D$, and it belongs to the closed unitary ball of $L^2(D)$, so that the optimization problem is well-posed. The sparsity of the method is ensured by the fact that $w$ must be zero on a set of measure $m$. As for the weighting vector in finite dimensions, the role of $w$ is to select the relevant features for clustering: if a Borel set $B \subset D$ is not relevant for cluster identification, we should expect $w(x) =0$ for $x \in B$, while if curves belonging to different clusters differ greatly in $B$, then $w$ should be strictly positive on that subset, with each value \(w(x)\) reflecting the importance of $x\in D$ for partitioning the data.

As for the finite dimensional case, we first focus on the problem of optimizing (\ref{eq:infinitesparsepb}) with respect to $w$ by holding $C_1,\ldots,C_K$ fixed. Problem (\ref{eq:infinitesparsepb}) can be tackled via an iterative algorithm, that alternatively computes the optimal $w$ by holding $C_1,\ldots,C_K$ fixed, or finds the best partition $C_1,\ldots,C_K$ by applying a suitable functional clustering technique, where the distance among functions is weighted according to the optimal $w$. We will give details on this clustering procedure in Section \ref{sec:algo}. We now focus on the following variational problem
\begin{align}\label{eq:infinitesparsepb_solow}
 \max_{w \in L^2(D)} & \int_{D} w(x) b(x) dx, \\
 \text{subject to:}\ \ & \Vert w(x) \Vert_{L^2(D)} \leq 1,\ \ w(x) \geq 0\  \mu-a.e.\ \text{and}\ \mu(\left\lbrace x \in D: w(x)=0\right\rbrace ) \geq m, \nonumber
\end{align}
where $0 < m < \mu(D) < \infty$, and $b(x)$ is continuous and non-negative a.e. Since $D$ is compact, $b\in L^\infty(D)$. The function $b(x)$ depends on the data and on the partition $C_1, \ldots  C_K$, and it defines how well the partition discriminates among clusters. Our main result states the existence and uniqueness of the solution to problem (\ref{eq:infinitesparsepb_solow}).
\begin{thm}\label{teo_infinito_dim}\ \\
There exists a solution to problem (\ref{eq:infinitesparsepb_solow}) given by $w(x) = \frac{b(x)}{\Vert b(x) \Vert_{L^2(B)} }I_B(x),$ where $I_B(x)$ is the indicator function of the set $B = \left\lbrace x \in D: b(x) > k \right\rbrace $, for a suitable $k \geq 0$. Moreover the solution is [$\mu$]-a.e. unique if the function $\phi(t):=\mu(\{x \in D:b(x)<t\})$ is continuous.
\end{thm}

\proof \ Let $\{b_n(x)\}_{n \in \mathbb{N}}$ be a monotone increasing sequence of simple functions $b_n(x) = \sum_{i = 1}^n c^n_i I_{B^n_i}(x)$, such that $b_n(x)$ converges to $b(x)$ in $L^{\infty}(D).$ The collection of sets $\{B_i^n\}_{i=1}^n$ forms a partition of $D,$ i.e. $B_i^n \cap B_j^n = \emptyset\ \forall\ i\neq j$ and $\bigcup_{i = 1}^n B^n_i = D$. Fix $k \in \mathbb{R}^+$ such that $\mu(\left\lbrace x \in D: b(x) > k \right\rbrace ) = \mu(D) - m.$ This is always possible thanks to the continuity of $\phi(t)$. Then, the set $A_k = \left\lbrace x \in D: b(x) \leq k \right\rbrace$ is such that $\mu(A_k) = m.$ Thanks to the continuity of $b(x)$ and the finiteness of $D$, $A_k$ is also a compact set. Let $H^n_k$ be the sub-collection of indices of the sets $\{B_i^n\}_{i=1}^n$ such that the intersection with $A_k$ is non empty $H^n_k:=\{h\in\{1,\ldots,n\}:A_k \cap B_h^n \neq \emptyset\},$ and let the number of elements in $H^n_k$ be $m^n_k$. Then, define $\tilde{B}^n_k := \cup_{h \in H^n_k} B_h^n.$ We have:
\begin{enumerate}
\item $b(x) > k, \forall x \in D\setminus \tilde{B}^n_k;$
\item $A_k \subseteq \tilde{B}^n_k,$ and thus $\mu(\tilde{B}^n_k) \geq \mu(A_k) = m$;
\item $c^n_h \leq k, \forall h \in H^n_k$.
\end{enumerate}
Consider $\mathbf{w},\mathbf{b}^n\in\mathbb{R}^n,$ and the following finite-dimensional optimization problem:
\begin{equation}\label{max discreto}
\max_{\mathbf{w}}  \left\langle \mathbf{w}, \mathbf{b}^n \right\rangle  
\end{equation}
\begin{equation*}
\text{s. t.}\ \Vert \mathbf{w} \Vert_{\ell^2} \leq 1,\ \ w_i \geq 0\  \forall i = 1,\ldots,n\ \text{and}\ \mu^{\sharp}(\left\lbrace i: w_i = 0\right\rbrace ) \geq m^n_k,
\end{equation*}
where $\mathbf{b}^n = (c^n_1, \ldots, c^n_n)^\prime$. According to Theorem \ref{Teo_sparse_vectorial}, the solution is given by the vector $\mathbf{w}^{*,n} = (w_1^{*,n}, \ldots, w_n^{*,n})^\prime$, where:
\begin{equation}\label{eq:finitodimproof}
w^{*,n}_i=\left \{ \begin{array}{ll}
              \frac{c^n_i}{(\sum_{i \in \mathcal{I}^n_k} (c^n_i)^2)^{\frac{1}{2}}}\ \text{if}\ i \in \mathcal{I}^n_k,\\
              &\\
              0\ \text{otherwise},
              \end{array}
    \right.
\end{equation}
with $\mathcal{I}^n_k = \{1,\ldots,n\}\setminus H^n_k.$ It is always possible to select the collection $\{B_i^n\}_{i=1}^n$ such that there are no ties in $\mathbf{b}^n$, and thus the solution in (\ref{eq:finitodimproof}) is unique for every $n$. Now, consider the sequence of functions $\left\lbrace w_n(x) \right\rbrace_{n \in \mathbb{N}}$ defined as $w_n(x) = \sum_{i = 1}^n w_i^{*,n}I_{B^n_i}(x)$. We claim that $w_n(x) \longrightarrow w(x)$ uniformly, where $w(x) = \frac{b(x)}{\Vert b(x)I_B(x) \Vert_{L^2(D)}}I_B(x)$ and $B = \left\lbrace x \in D: b(x) > k \right\rbrace.$ Indeed, because of the uniform convergence of $b_n(x)$ to $b(x)$ and the boundedness of $D$, we have that
\begin{equation}\label{unif conv}
\sum_{i \in \mathcal{I}^n_k} c^n_iI_{B^n_i}(x) \xrightarrow{L^{\infty}} b(x)I_B(x)
\end{equation}
and $\lVert b_n(x) \lVert_{L^p} \rightarrow \lVert b(x) \lVert_{L^p},\ \forall p\geq 1.$ Then
\begin{eqnarray}\label{eq:closingproof}
\lVert w_n(x) - w(x) \lVert_{L^{\infty}} & \leq & \Vert  \sum_{i} \frac{c^n_i}{(\sum_{i \in \mathcal{I}^n_k} (c^n_i)^2)^{\frac{1}{2}}}I_{B^n_i}(x) - \sum_{i}\frac{c^n_i}{\lVert b(x) \rVert_{L^2(B)}}I_{B^n_i}(x)  \Vert_{L^{\infty}} \nonumber\\
& + & \Vert \sum_{i}\frac{c^n_i}{\lVert b(x) \rVert_{L^2(B)}}I_{B^n_i}(x) - \frac{b(x)}{\Vert b(x) \Vert_{L^2(B)}}I_B(x)  \Vert_{L^{\infty}}.
\end{eqnarray}
As $n\rightarrow\infty$, the second term in (\ref{eq:closingproof}) vanishes because of (\ref{unif conv}), whereas for the first term we have:
\begin{equation*}
\Vert  \sum_{i} \frac{c^n_i}{(\sum_{i \in \mathcal{I}^n_k} (c^n_i)^2)^{\frac{1}{2}}}I_{B^n_i}(x) - \sum_{i}\frac{c^n_i}{\lVert b(x) \rVert_{L^2(B)}}I_{B^n_i}(x)  \Vert_{L^{\infty}} \leq
\end{equation*}
\begin{equation*}
 \max_{x \in D} b(x) \cdot |\frac{\lVert b(x) \lVert_{L^2} - \lVert b_n(x) \lVert_{L^2}}{\lVert b_n(x) \lVert_{L^2}\lVert b(x) \lVert_{L^2}}| \leq \frac{\max_{x \in D} b(x)}{\lVert b_1(x) \lVert^2_{L^2}} \cdot |\lVert b(x) \lVert_{L^2} - \lVert b_n(x) \lVert_{L^2}|,
\end{equation*}
which tends to 0, when $n \rightarrow \infty$, given that $b_1(x) \neq 0$. If $b_1(x) \equiv 0$, then we take $b_2(x)$ in the last inequality and the result still holds. It is easy to check that $w(x)$ satisfies the norm and the non-negativity constraints. The constraint on the measure of the sparsity set is implied by the uniform convergence and the maximality of $w(x)$ follows from the continuity of the limit. The uniqueness a.e. of the solution is obvious.
\endproof

Since the proof makes use of Theorem \ref{Teo_sparse_vectorial} in the simple function approximation, we do not encounter the same issues about the uniqueness of the solution which we were facing in the finite dimensional case. Indeed, if the function $b(x)$ is constant on an interval of positive measure, it is sufficient to keep the partition $\{B_i^n\}_{i=1}^n$ fixed $\forall\ n$ on the interval where $b(x)$ is constant, and let it vary on the rest. We remark that an alternative proof of this result, which completely relies on variational theory, is reported in the Appendix.

\subsection{A K-means implementation of Functional Sparse Clustering}\label{sec:algo}

We here propose a possible iterative algorithm implementing the functional sparse clustering framework described so far. In particular, we aim at solving the variational problem in (\ref{eq:infinitesparsepb}), where $g(\cdot)$ is specified by (\ref{eq:KmeansFunc}), and thus the solution $w:D \rightarrow \mathbb{R}$ quantifies the increase in the BCSS that each portion of the domain $D$ could generate. Note that we have to maximize over the infinite dimensional set of weighting functions satisfying the constraints, and jointly on the (finite, but possibly huge) set of all possible cluster assignments.

To this aim, we propose the following $K$-mean inspired iterative strategy, that we name \emph{Functional Sparse $K$-means Clustering Algorithm:}
\begin{enumerate}
\item initialize $C_1,\ldots,C_K$ by running functional $K$-means \citep{tk};
\item iterate until convergence:
\begin{itemize}
\item[(i)] holding $C_1,\ldots,C_K$ fixed, use Theorem \ref{teo_infinito_dim} to obtain the optimal weighting function $w^*(x)$;
\item[(ii)] holding $w=w^*$ fixed, find the optimal partition $C_1,\ldots,C_K$ by optimizing criterion (\ref{eq:infinitesparsepb}) with BCSS, i.e., run a functional $K$-means algorithm according to the weighted measure
\begin{equation}\label{eq:wmeasure}
d_w(f_i,f_{i^\prime}) = \int_D w(x)(f_i(x)-f_{i^\prime}(x))^2dx;
\end{equation}
\end{itemize}
\item stop when there are no more changes in the partition.
\end{enumerate}
The optimal weighting function is given by the $w^*$ obtained at the last (i) step, while the final clusters are given by $C_1,\ldots,C_K$ in the last (ii) step. As we already noted in the multivariate case, this kind of procedure only assures that the value of the objective functional is increased at each step, but it does not assure that a global optimum is achieved. We also remark that the same permutation-based approach for tuning $m$ on the basis of a GAP statistics, proposed for the finite dimensional setting, can be applied here: the domain $D$ is subdivided into a possibly large set of sub-domains, and the procedure is then applied to the observed functions after they have been randomly permuted within each sub-domain. The number of sub-domains depends obviously on the available computing power and on the data dimension.

\section{Simulation Studies}\label{sec:simulations}

The previously described setting for sparse clustering is here tested on synthetic data in various situations.

\subsection{A simulation study for multivariate sparse clustering}\label{subsec:sim1}

Aim of the present simulation study is to compare results of our $K$-means implementation of multivariate sparse clustering with results given by standard $K$-means and by the sparse $K$-means proposed in \cite{twitt}.

The simulation is run as follows. The number of classes is set to $K = 3$, both for data generation and for clustering, and we generate $20$ observations per class, thus leading to $n = 60$ data in each scenario. Only $q = 10$ features are responsible for differences among classes. We simulate three different scenarios, where the data dimension $p$ is respectively equal to $50,200,500$: this means that in the less extreme situation $20\%$ of the features are relevant, while in the more extreme one only $2\%$ of them are relevant. Data are independently generated from Gaussian distributions such that, for $i=1,\ldots,n$ and $j=1,\ldots,p$, $X_{ij} \sim N(\mu_{ij},\sigma^2)$, with $\sigma=0.2$ and 
$$
\mu_{ij} = \left\{\begin{array}{ll}
j/p & \text{if }j>q;\\
j/p + 1.5\sigma(1_{C_2}(i) - 1_{C_3}(i)) & \text{if }j\leq q;
\end{array}\right.
$$
where $C_k,\ k=1,\ldots,K$ is the set of data belonging to the $k$-th cluster and $1_A(\cdot)$ is the indicator function of the set $A$. In order to compare the different partitions we use the \emph{Classification Error Rate} (CER), which equals 0 if the partitions agree perfectly, and 1 in case of complete disagreement. Note that CER also equals 1 minus the Rand index \citep{Rand71}.

\begin{table}
  \centering
  \footnotesize{
  \begin{tabular}{cccc}
           & $p = 50$ & $p = 100$ & $p = 500$ \\
    \hline
           std        &  0.0474(0.0641)  &    0.121(0.0640) & 0.213(0.0630)  \\
           WT sparse  &  0.0277(0.0295)  &    0.0216(0.0219) & 0.0307(0.0281)  \\
           NEW sparse &  0.0106(0.0106)  &    0.0118(0.0117) & 0.0225(0.0273)  \\
    \end{tabular}}
    \caption{results of the simulation study of Section \ref{subsec:sim1}. Mean (standard devation) of the CER over 20 simulations. Along rows, results obtained with the three methods: standard $K$-means (std), sparse $K$-means of \cite{twitt} (WT sparse) and our implementation of sparse $K$-means (NEW sparse).}\label{tab:sim1} 
\end{table}

\begin{figure}[t]
\centerline{\includegraphics[width=.7\textwidth]{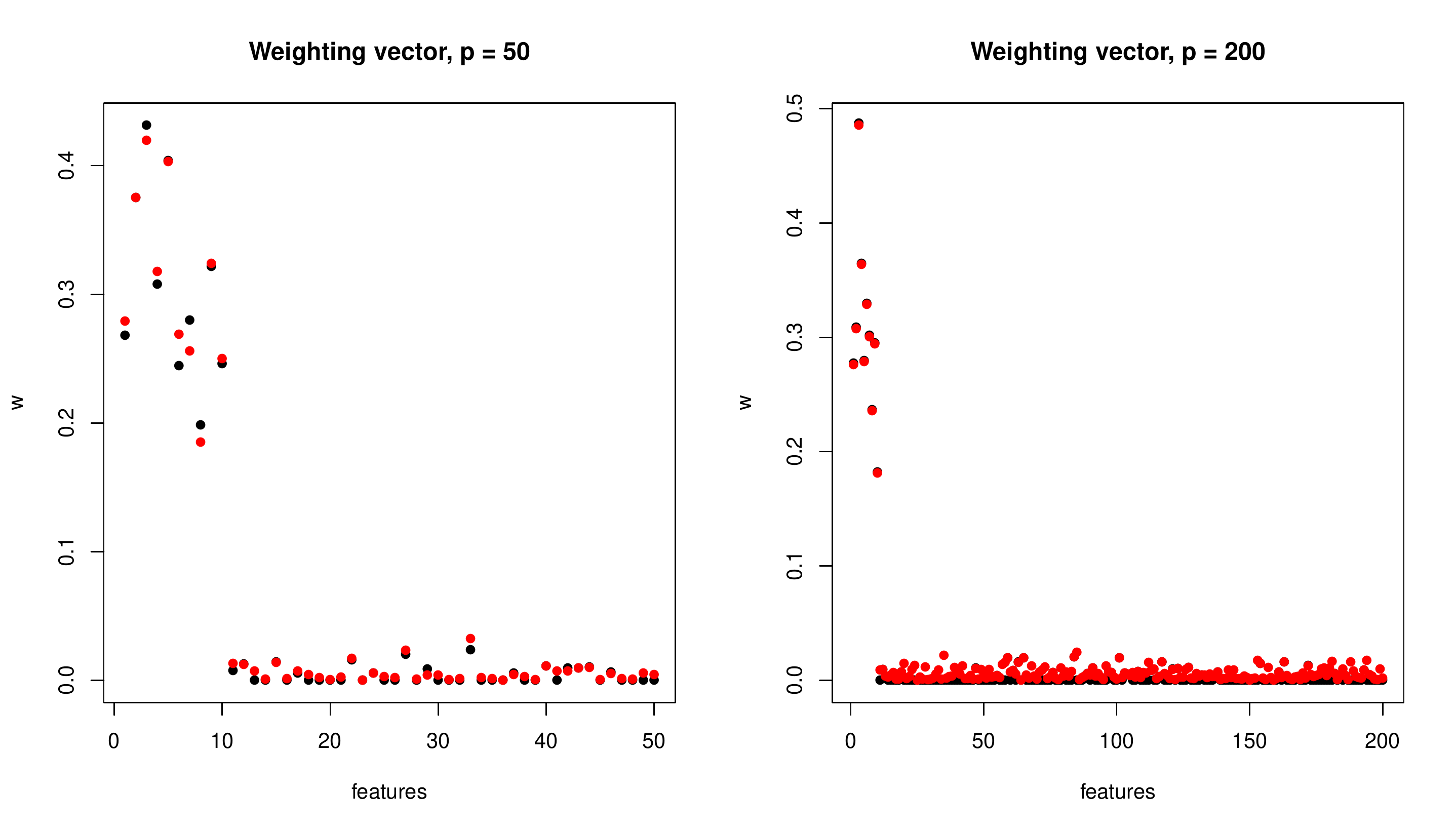}}
\caption{results of the simulation study of Section \ref{subsec:sim1}. Plot of the weighting vector $\mathbf{w}$ obtained with the sparse $K$-means of \cite{twitt} (red), and with our implementation of sparse $K$-means (black), on two synthetic datasets with $p = 50$ (left) and $p = 200$ (right).}\label{fig:sim1}
\end{figure}

We aim at checking the adherence of the detected grouping to the real known classes: we expect both our sparse $K$-means and the proposal in \cite{twitt} to be accurate in all situations, and to improve on standard $K$-means if the proportion of redundant features is large. Results are shown in Table \ref{tab:sim1}: we notice that both sparse methods perform well in all situations, while the performance of $K$-means becomes worse when a large proportion of noisy irrelevant features is present. We also notice a slightly better performance of our implementation of sparse $K$-means with respect to \cite{twitt}. The slight superiority of our sparse approach might be explained in the light of a more direct impact of the sparsity parameter $m$ on clustering results, with respect to the $s$ parameter of \cite{twitt}. In Figure \ref{fig:sim1} we compare the two methods in terms of the optimal weighting vector $\mathbf{w}$, for two simulations with $p = 50$ (left) and $p = 200$ (right): in our implementation of sparse $K$-means, the optimal sparse parameter estimated via GAP statistics is $m = 25$ and $m = 160$ for the two simulations, respectively. This means that 50\% and more than 75\% of the features, respectively in the two cases, are correctly recognized as irrelevant and thus discarded from the classifier with a zero associated weight. The result obtained with the approach described in \cite{twitt}, instead, shows that a positive small weight is given to all features, thus potentially including noisy confounders in the classifier.

In conclusion, it is worth remarking that the two procedures are comparably accurate, since both correctly assign larger weights to the first 10 features, the ones that are responsible for differences in the signal distribution across classes. In our approach, though, the sparse parameter has a more straightforward interpretation in terms of the number of irrelevant features, which are immediately pointed out by a null weight.

\begin{figure}[t]
\centering\includegraphics[width=0.4\textwidth]{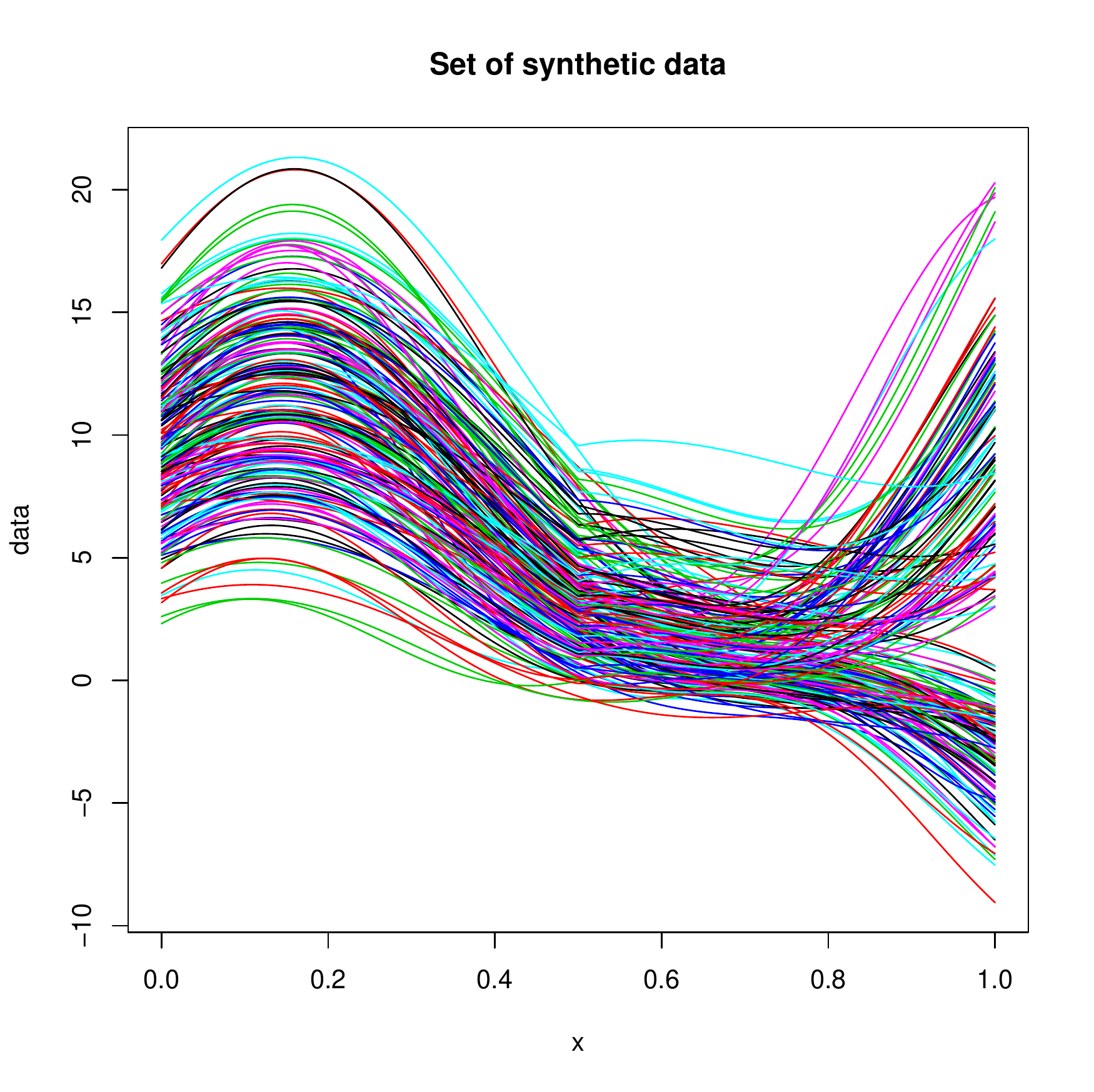}
\caption{sample of synthetic functional data generated for the simulation study in Section \ref{subsec:sim2}.}\label{fig:sim2data}
\label{last}
\end{figure}

\begin{figure}[!h]
\centering\includegraphics[width=.8\textwidth]{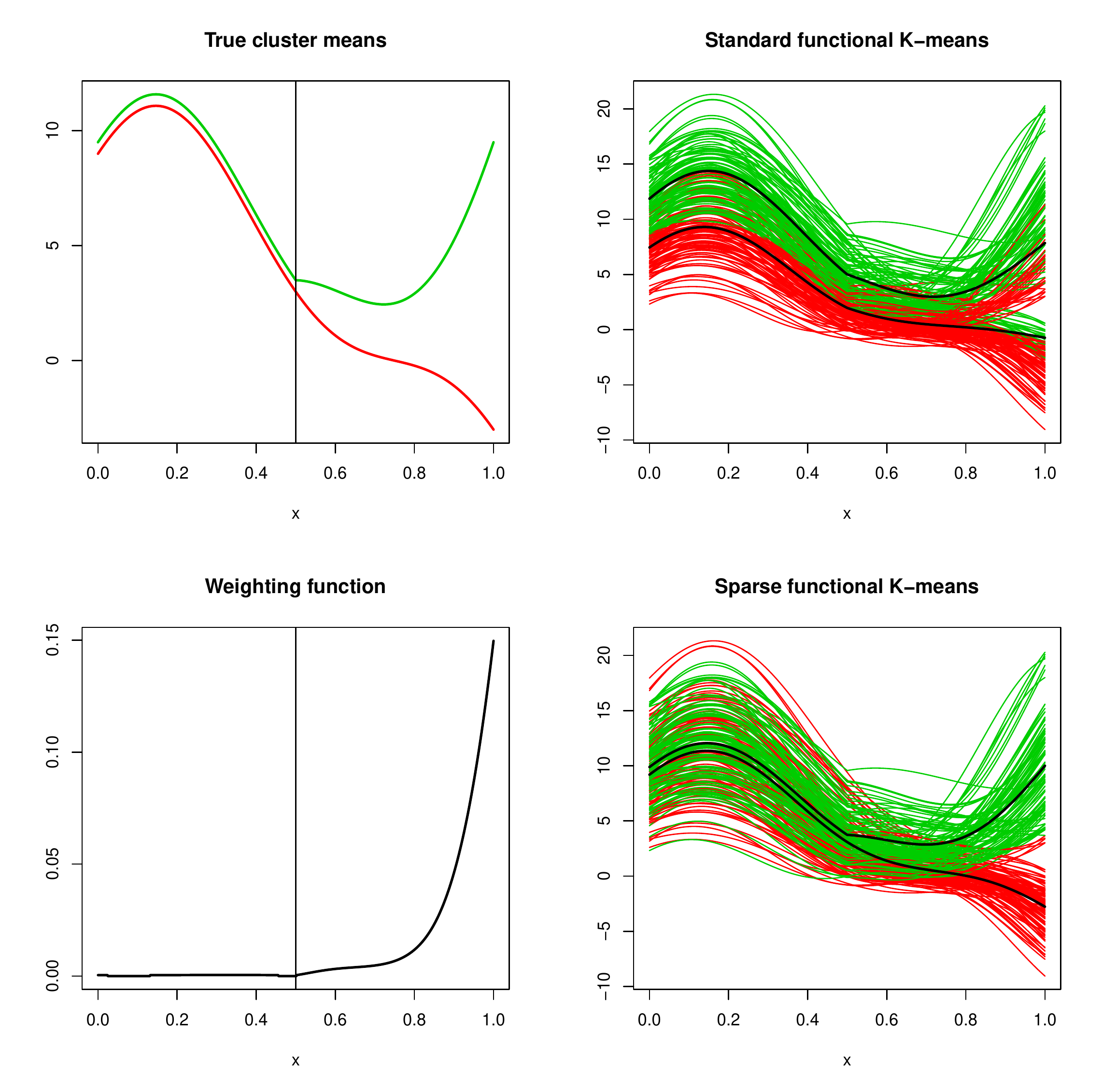}
\caption{results of the simulation study of Section \ref{subsec:sim2}. True cluster mean functions (top, left); one of the synthetic datasets coloured according to the clusterization obtained with standard functional $K$-means (top, right) and with sparse functional $K$-means (bottom, right); optimal weighting function computed by the sparse approach (bottom, left). The vertical line in the left plots is drawn for $x=1/2$.}\label{fig:sim2}
\end{figure}

\subsection{A simulation study for functional sparse clustering}\label{subsec:sim2}

Aim of this second simulation study is to test our innovative $K$-means implementation of functional sparse clustering. Results will be compared with standard functional $K$-means.

The simulation is run as follows. The number of classes is set to $K = 2$, both for data generation and for clustering, with $100$ observations per class, thus leading to $n = 200$ data in each scenario. The domain is chosen to be the interval $D=[0,1]$. The mean function of data belonging to the first cluster is $f_1(x) = (b\sin(b\pi x)+a)\cdot(a-4x) + c,\quad x\in D,$ with $a=3$, $b=2$ and $c=0$, while the second cluster has mean function
\begin{equation}\label{eq:media2}
f_2(x) = \left\{\begin{array}{lc}
(b\sin(b\pi x)+a)\cdot(a-4x)+c    & x\in[0,\frac{1}{2}];\\
(b\sin(b\pi x)+a)\cdot (a-4(1-x))-2c(x-1) & x\in(\frac{1}{2},1],
\end{array}\right.
\end{equation}
with $c=1/2$. The two true cluster mean functions are shown in Figure \ref{fig:sim2} (top-left panel): note that they are nearly equal on the first half of the domain, the only difference being a vertical shift of 1/2. On the second portion of the domain, instead, they become increasingly different for values of the abscissa $x$ closer to 1. Therefore we expect the weighting function to give more importance to the second half of the domain, with higher values when closer to 1, and to be zero in the first half. Data are generated according to $f_1(x)$ and $f_2(x)$, for the first and second cluster respectively, with $a \sim N(3,0.5^2)$ and $b \sim N(2,0.25^2)$. For data belonging to the first cluster, the mean function $f_1(x)$ is used with $c \sim N(0,0.5^2)$, while for data belonging to the second cluster $f_2(x)$ is used with $c \sim N(0.5,0.5^2)$. Note that a noise term has not been added to the synthetic data: indeed, one of the assumptions of our proposal is that the functional form of the data has already been reconstructed (i.e., no smoothing is needed). One of the sets of synthetic data analysed in this simulation study is shown in Figure \ref{fig:sim2data}: data are completely overlapped for $x\in[0,1/2]$, but cluster overlap is still consistent also for $x>1/2$; a sparse clustering approach is thus needed.

\begin{table}
  \centering
  \footnotesize{
  \begin{tabular}{ccccccccccc}
         run  & 1 & 2 & 3 & 4 & 5 & 6 & 7 & 8 & 9 & 10 \\
    \hline
           std     & 0.361 & 0.387 & 0.356 & 0.434 & 0.422 & 0.437 & 0.482 & 0.382 & 0.401 & 0.333 \\
           sparse  & 0.0345 & 0.122 & 0.0567 & 0.113 & 0.0490 & 0.104 & 0.0502 & 0.0346 & 0.104 & 0.0611 \\
    \end{tabular}}
    \caption{results of the simulation study of Section \ref{subsec:sim2}, showing the CER values over 10 simulations. Along rows, results obtained with the two competing methods: standard functional $K$-means (std), and sparse functional $K$-means (sparse).}\label{tab:sim2} 
\end{table}

First, the CER index is used to check the adherence to the real known classes of the grouping structure detected by the two clustering strategies we are comparing: after 10 repetitions of the simulation, a standard functional $K$-means provides an average CER of 0.3996, while sparse functional $K$-means results in 0.07306, a reduction of more than 80\%. The CER values for all simulations are shown in Table \ref{tab:sim2}: in all simulated scenarios, the sparse approach manages to attain a far better adherence to the real underlying grouping structure.

Secondly, we aim at extensively commenting the results of one simulation, to check the misclassification error of the two procedures, the estimated cluster mean functions, and in the case of sparse functional $K$-means, to inspect the estimated weighting function $w$. Results obtained in the fifth run of the simulation, whose CER value is shown in the fifth column of Table \ref{tab:sim2}, are shown in Figure \ref{fig:sim2}. As it can be appreciated from the top-right panel in the picture, the standard functional $K$-means does not provide a meaningful result: nearly one third of the curves are misclassified, and this also reflects on the poor estimation of cluster mean functions shown in black in the plot (they can be compared with the true cluster means, shown in the top-left panel of the same figure).  Results obtained with sparse functional $K$-means, instead, show a pretty nice cluster assignment, coherent with the true underlying grouping structure, and a very good estimation of the cluster mean functions.

\begin{table}
  \centering
  \footnotesize{
\begin{tabular}{cc|cc}
std 2-means & & \multicolumn{2}{c}{clusters} \\
\cline{3-4} & & \multicolumn{1}{c}{1} & 2  \\
\cline{1-4} \multicolumn{1}{c|}{true} & 1 & 73 & 27 \\
 & \multicolumn{1}{|c|}{2} & 33 & 67 \\
\end{tabular}
\hspace{1.5cm}
\begin{tabular}{cc|cc}
sparse 2-means & & \multicolumn{2}{c}{clusters} \\
\cline{3-4} & & \multicolumn{1}{c}{1} & 2  \\
\cline{1-4} \multicolumn{1}{c|}{true} & 1 & 100 & 0 \\
 & \multicolumn{1}{|c|}{2} & 5 & 95 \\
\end{tabular}}
    \caption{results of the simulation study of Section \ref{subsec:sim2}, showing the confusion matrix of cluster assignments vs true labels. Left: results of standard functional 2-means clustering; right: results of sparse functional 2-means clustering.}\label{tab:sim2bis} 
\end{table}

The confusion matrices of the two partitions obtained with standard and sparse functional $K$-means with respect to the true labels are shown in Table \ref{tab:sim2bis}: with standard $K$-means, 30\% of the data are misclassified, while only 2.5\% with sparse functional $K$-means.

Finally, let us consider the weighting function estimated by sparse functional $K$-means, and shown in the bottom-left panel of Figure \ref{fig:sim2}: the weighting function is non-zero only on the subinterval of the domain $[0.521,1]$, and the optimal sparsity parameter is $m = .521$. This is obviously pointing out the fact that the true cluster mean functions differ mostly in the second half of the domain, thus detecting precisely the region of maximal cluster distinction. Moreover, $w$ is a strict monotone increasing function on $[0.521,1],$ thus also suggesting that data belonging to different clusters distinguish most for values of the abscissa $x$ closer to 1, as it is evident from the top-left plot in Figure \ref{fig:sim2}.

\section{Case Study: Berkeley Growth Data}\label{sec:realdata}

In this section we illustrate the results obtained with sparse functional clustering on the growth curves included in the Berkeley Growth Study, a benchmark dataset for functional data analysis which is also provided in the \verb"fda" package \citep{fda-package} in \verb"R" \citep{R}. Aim of this Section is to check whether a sparse approach improves the insight on these data with respect to standard techniques, and whether the weighting function $w$ helps in the interpretation of the results. Moreover, we aim at deepening the discussion on the role of the sparsity parameter, $m$, in the context of a real application. 

\begin{figure}[t]
\begin{center}
\centering\includegraphics[width=.9\textwidth]{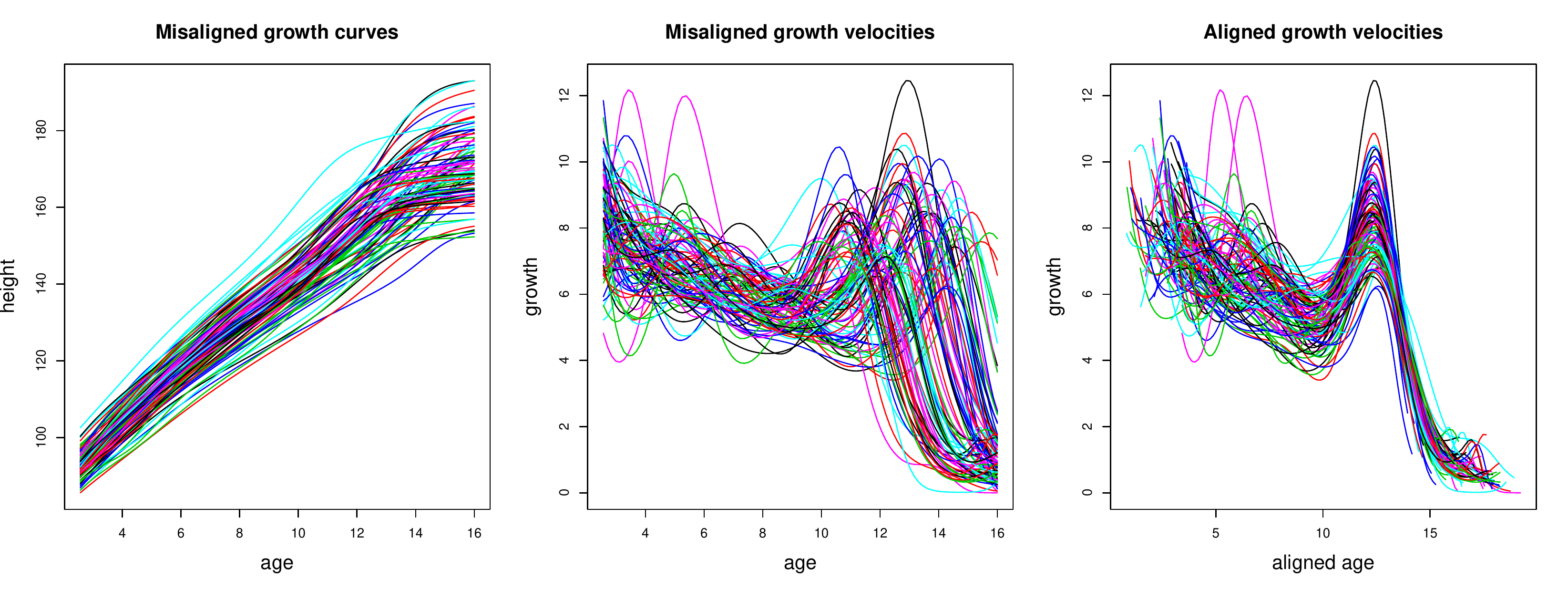}
\end{center}
\caption{the Berkeley Growth Study dataset. From left to right: reconstructed growth curves, reconstructed growth velocities, and aligned growth velocities.}\label{fig:growth_data}
\end{figure}

The Berkeley Growth Study \citep{Tuddenham-Snyder-1954} is one of several long-term developmental investigations on children conducted by the California Institute of Child Welfare. It includes the heights (in cm) of 93 children, 54 girls and 39 boys, measured quarterly from 1 to 2 years, annually from 2 to 8 years and then biannually from 8 to 18 years. It is reasonable to consider these data as discrete observations of a continuous process representing the height of each single child, i.e. the child growth curve. We reconstruct the growth curves by means of monotonic cubic regression splines \citep{rs2}, implemented using the \texttt{R} function \texttt{smooth.monotone} available in the \texttt{fda} package \citep{fda-package}. Estimated growth curves and velocities are shown in Figure \ref{fig:growth_data} (left and central panels, respectively).

We aim at comparing the results obtained via sparse functional $K$-means with the ones obtained with standard functional $K$-means. As suggested in many papers analysing the same dataset, we will focus on growth velocities. From inspection of the central panel in Figure \ref{fig:growth_data}, we notice that all children show a similar growth pattern, characterized by the pubertal spurt, known in the medical literature as a sharp peak of growth velocity between 10 and 16 years. However, the children follow their own biological clocks, thus resulting in a set of velocity curves showing the main growth spurt with different timing/duration. Moreover, a second minor feature emerges: some children also have a minor growth velocity peak between 2 and 5 years, called mid-spurt, which can be quite consistent but also absent. We would also like this minor feature not to be confounded by phase variability. Hence, we also analyse an aligned version of growth velocities, obtained by applying 1-mean alignment: this technique, first introduced in \citet{ssvv}, has been tested on the Berkeley Growth Study, thus showing that  it effectively decouples amplitude and phase variability. Aligned growth velocities are also shown in Figure \ref{fig:growth_data} (right panel).

\begin{table}
  \centering
  \footnotesize{
  \begin{tabular}{cc|cc}
std 2-means & & \multicolumn{2}{c}{clusters} \\
\cline{3-4} & & \multicolumn{1}{c}{1} & 2  \\
\cline{1-4} \multicolumn{1}{c|}{gender} & M & 37 & 2 \\
 & \multicolumn{1}{|c|}{F} & 9 & 45 \\
\end{tabular}
\hspace{1.5cm}
\begin{tabular}{cc|cc}
sparse 2-means & & \multicolumn{2}{c}{clusters} \\
\cline{3-4} & & \multicolumn{1}{c}{1} & 2  \\
\cline{1-4} \multicolumn{1}{c|}{gender} & M & 37 & 2 \\
 & \multicolumn{1}{|c|}{F} & 9 & 45 \\
\end{tabular}}
    \caption{results of the analysis of the Berkeley Growth Study data, showing the confusion matrix of cluster assignments vs gender. Left: results of standard functional 2-means clustering; right: results of sparse functional 2-means clustering}\label{tab:berkeley} 
\end{table}

As it can be appreciated in Table \ref{tab:berkeley}, the results of sparse and standard functional $2$-means of misaligned growth velocities are exactly the same: the two partitions coincide, and both estimate the gender classification quite well. Moreover, the weighting function estimated by sparse functional 2-means (Figure \ref{fig:growth_dataMF}, left panel) mainly points out the pubertal peak period. All these conclusions are reasonable, and quite coherent with previous findings \citep{ssvv,Tarpey14}. But can we gain some further insight when analysing the aligned growth velocities sample, where phase variability has been properly removed?

\begin{figure}[!ht]
\centering\includegraphics[width=.7\textwidth]{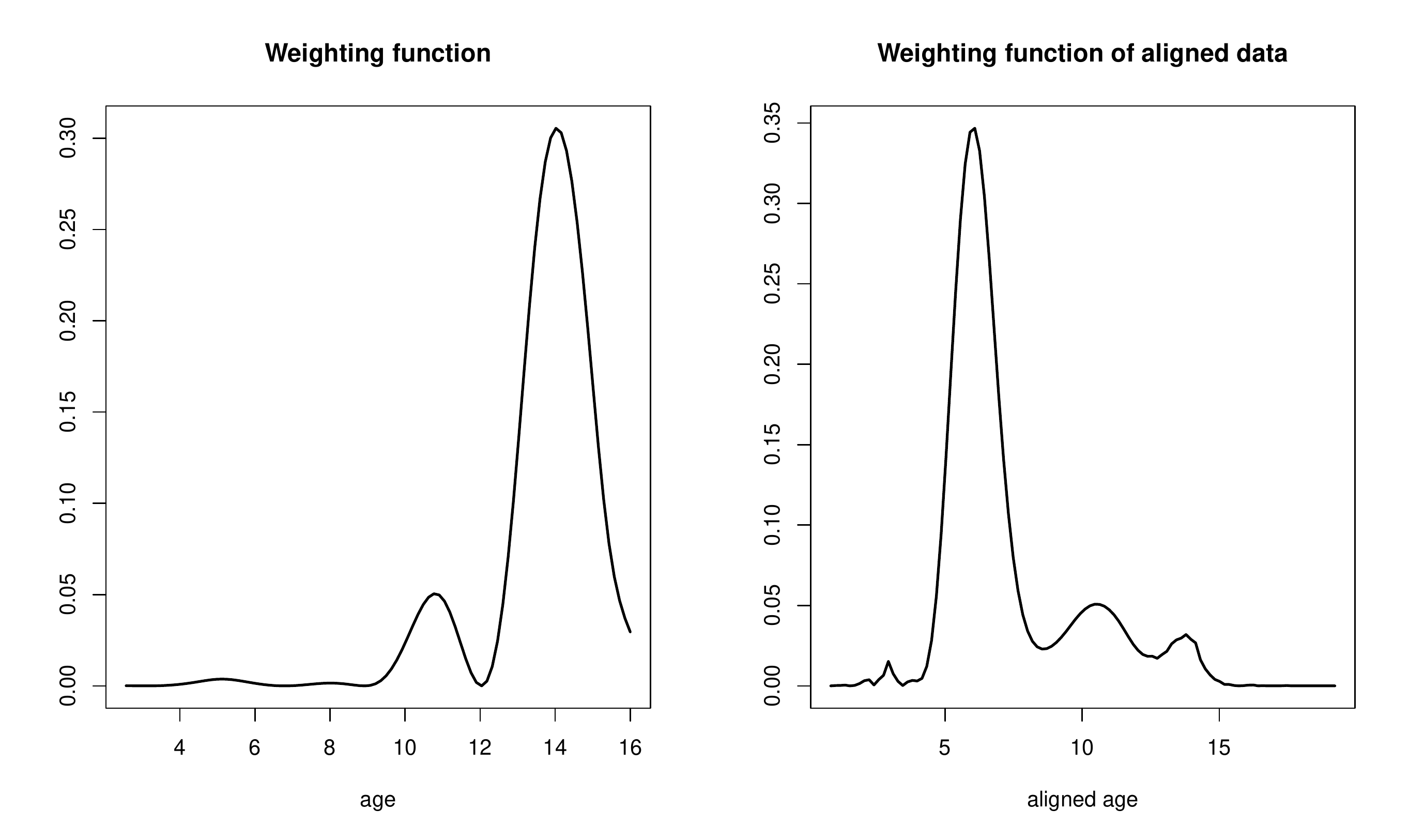}
\caption{results of the analysis of the Berkeley Growth Study data. Estimated weighting functions obtained by applying sparse functional 2-means to misaligned growth velocities (left) and to aligned ones (right).}\label{fig:growth_dataMF}
\end{figure}

Interestingly, when analysing the aligned growth velocities with standard and sparse functional 2-means, we obtain two quite different grouping structures, none of which reflects gender stratification. The gender and standard 2-means classifications seem not to point out relevant features characterizing the clusters (see Figure \ref{fig:growth_ZOOM}, center and bottom). Instead, the 2 clusters obtained by sparse functional 2-means distinguish among children having their mid-spurt, from children not having it (see Figure \ref{fig:growth_ZOOM}, top). This fact is even more evident if we look at the weighting function estimated by sparse functional 2-means when applied to aligned growth velocities, shown in Figure \ref{fig:growth_dataMF} (right panel): the period of the mid-spurt is there indicated as the most relevant portion of the domain.

In conclusion, the sparse functional $K$-means clustering of the aligned growth velocities shows novel insight on the Berkeley Growth Study data, finally pointing out a data stratification possibly related to the children's development in the early stages of their lives.

\clearpage
\newpage

\begin{figure}[!ht]
\centering\includegraphics[width=.7\textwidth]{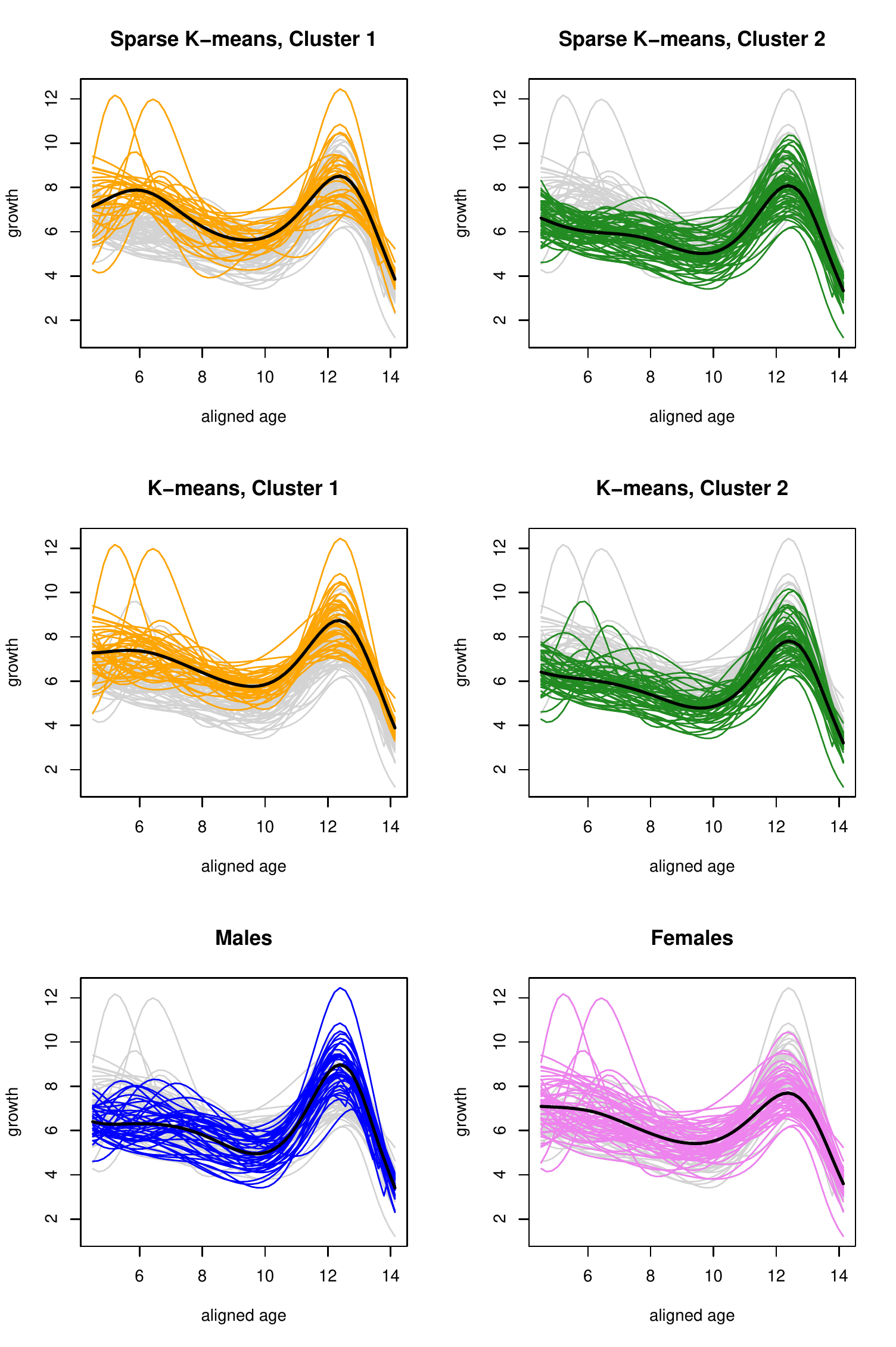}
\caption{in grey in all panels, aligned growth velocities restricred to the subset of the domain where the weighting function estimated by sparse functional 2-means has higher values. In the top panels, curves are colored according to the clustering given by sparse functional 2-means. In the center panels, they are colored according to the result of standard functional 2-means, while in the bottom panels they are colored according to gender (blue for males and pink for females).}\label{fig:growth_ZOOM}
\end{figure}

\section{Discussion}\label{sec:discussion}

In this paper we have first proposed a framework for sparse clustering of multivariate data based on hard thresholding, and proved its accuracy on simulated scenarios. We have then started from this novel definition of sparse clustering for multivariate data to properly define a quite general framework for sparse clustering of functional data. We have proven the existence and uniqueness of the optimal solution, and proposed a quite general and flexible algorithm to estimate it. Finally, we have tested the method on both simulated and real data.

Given the high level of novelty and generality of the paper, there are many further directions arising from the present proposal which are worth being explored. Firstly, properly treating the possible misalignment in the data is an issue: the problem of decoupling phase and amplitude variability is often encountered in functional data analysis \citep{rs2}, and we indeed face it in the applied case study of Section \ref{sec:realdata}. It the context of $K$-means clustering it has been considered in \cite{ssvv}, where the authors prove the optimality of performing clustering and alignment jointly. It would be then interesting and relevant to develop a joint sparse functional clustering and alignment method, possibly merging the two approaches.

Another interesting open issue concerns the functional measure: the variational problem (\ref{eq:infinitesparsepb}) could have a different solution if another functional measure different from $L^2$ were to be considered, but we may then also not be able to prove results similar to Theorem \ref{teo_infinito_dim}. An interesting choice would be, for instance, the measure of the density generating the data. We would then have the following formulation for problem (\ref{eq:infinitesparsepb})
\begin{equation}\label{nf}
 \max_{w(x), (C_1,\ldots,C_K)} \int_D w(x)b(x, f_{i,k}, \mathbf{C})d\mathbb{P}(\mathcal{X}_k(\omega, x), x),
\end{equation}
where $\mathbb{P}(\mathcal{X}_k(\omega, x), x)$ is a probability measure depending both on the functional variables generating the data and the domain, while $\mathbf{C} = (C_1,\ldots,C_K)$ is the partition. A related possibility is to cluster the data jointly with their successive derivatives, if they exist. Again, the problem is proving the existence/uniqueness of the solution in the proper functional space.

Another issue related to the variational problem concerns the optimization strategy: it would be interesting to define an iterative procedure converging to a global maximum. Moreover, if not solved in the previous point, it would be interesting also to try other clustering strategies. Finally, the tuning parameter $m$ is also worth some further thoughts: the GAP statistics-inspired criterion that we are currently using is not really generalized but just adapted to the functional case, and surely something better can be developed.

\section*{Acknowledgement}
We wish to thank Prof. Piercesare Secchi (MOX, Department of Mathematics, Politecnico di Milano, Italy) for motivating us on the study of this challenging problem, for the invaluable help on the theoretical results, and for being and extraordinary PhD advisor for (nearly) both of us. Moreover, we thank Prof. Douglas S. Bridges (Department of Mathematics and Statistics, University of Canterbury, New Zealand) for the fruitful discussions on the variational proof of Theorem \ref{teo_infinito_dim}.

\appendix
\section{Appendix: alternative proof of Theorem \ref{teo_infinito_dim}}

We here give a completely variational proof of Theorem \ref{teo_infinito_dim}.

\begin{thm}\label{teorema}

The variational problem:
\begin{align}\label{eq:infinitesparsepb_solow_new}
 \max_{w \in L^2(D)} & \int_{D} w(x) b(x) dx, \\
 \text{subject to:}\ \ & \Vert w(x) \Vert_{L^2(D)} \leq 1,\ \ w(x) \geq 0\  \mu-a.e.\ \text{and}\ \mu(\left\lbrace x \in D: w(x)=0\right\rbrace ) \geq m, \nonumber
\end{align}
with $D \subset \mathbb{R}$ compact, $0 < m < \mu(D) < \infty$, $b(x)$ continuous and non-negative a.e. and $\mu$ a regular, non-atomic and translation-invariant measure has a solution given by $w(x) = \frac{b(x)}{\Vert b(x) \Vert_{L^2(B)}}I_B(x)$, where $B = \{x \in D: b(x) > k \}$ and $k\in\mathbb{R}^+$ depend on $m$. Moreover, the solution is [$\mu$]-a.e. unique if the function $\phi(t):=\mu(\{x \in D:b(x)<t\})$ is continuous.

\end{thm}

\begin{proof} 

Let us fix a $k \in \mathbb{R}^+$, thanks to the continuity of $b$, such that $\mu(B) = \mu(D) - m$. It is clear that the function $w(x)$ as defined above, satisfies the constraints in (\ref{eq:infinitesparsepb_solow_new}). We claim that $w(x)$ is also the solution of (\ref{eq:infinitesparsepb_solow_new}).

We will proceed in the proof deriving absurd statements if we assume that the maximizing function differs somehow from the proposed solution $w(x)$. So, suppose $w(x)$ is not the maximum for the functional in (\ref{eq:infinitesparsepb_solow_new}). Then, there exists another function $g(x)$ that satisfies the constraints and
\begin{equation}\label{eq: disuguaglianza g}
\int_{D} g(x)b(x) d\mu > \int_{D} w(x)b(x) d\mu.
\end{equation} 
Define the set $G := \{x \in D: g(x) > 0 \}$. We will proceed firstly by showing that, if we suppose that $G$ differs from $B$ for a set of positive measure, we will always be able to find another function, more similar to the defined $w(x)$ and satisfying the constraints, that gives us a higher value in the functional in (\ref{eq:infinitesparsepb_solow_new}).  Then, given that $G = B$ $\left[  \mu \right]-$a.e., if we suppose that (\ref{eq: disuguaglianza g}) holds and $g(x) \neq w(x)$ on a set of positive measure, again we obtain contradictory statements.  

So, let us observe the following considerations.

$B \subset G$, with $\mu(B) < \mu(G)$ cannot happen, because $g(x)$ would not satisfy the measure constraint. 

It cannot be that $G \subset B$ and $\mu(G) < \mu(B)$, because, for any $g$ satisfying the constraints, we would have:
\begin{eqnarray}\label{primo caso}
\int_{D} g(x)b(x) d\mu & = & \int_{G} g(x)b(x) d\mu \\ \nonumber
& \leq & \Vert g(x) \Vert_{L^2(G)} \Vert b(x) \Vert_{L^2(G)} \\ \nonumber
& \leq & \Vert b(x) \Vert_{L^2(G)} \\ \nonumber
& \leq & \Vert b(x) \Vert_{L^2(B)} \\ \nonumber
& = & \int_{B} w(x)b(x) d\mu = \int_{D} w(x)b(x) d\mu.
\end{eqnarray}

Now, suppose that the sets $G$ and $B$ are different, with a non-empty intersection with each other. Specifically, let us suppose that
\begin{enumerate}
\item $G \cap B \neq \emptyset$;
\item $G \cap (D\setminus B) \neq \emptyset$;
\item $\mu(G \cap B) > 0$ and $\mu(G \cap (D\setminus B)) > 0$;
\item $\mu(G \cap B) + \mu(G \cap (D\setminus B))  \geq \mu(B)$.
\end{enumerate}
This is the most complicated case, therefore we initially assume $B \setminus (B \cap G)$ contains an interval $J$. We want to show that, given those hypotheses on the sets $B$ and $G$, \textit{no matter how} $g$ \textit{is}, we can \textit{build} a function $\tilde{g}$ for which we reach a greater value in the functional in (\ref{eq:infinitesparsepb_solow_new}). For the sake of notation, let us set $\Gamma := G \cap (D\setminus B)$. $\Gamma \subset D\setminus B$, which is compact, because $B$ is open in $D$. So, there exists an open finite subcover with intervals, $\bigcup_{i = 1}^n K_i$, of $D \setminus B$. There exists at least one of the $K_i$'s which has non empty intersection with $\Gamma$ and that intersection is of positive measure, otherwise $\Gamma$ itself is of measure zero, which would contradict the hypothesis. So, $\Gamma \subset \bigcup_{i \in \mathcal{I}} K_i$, with $\mathcal{I} \subset \{1, \ldots, n\}$ and take a $K_i$ (or a suitable subinterval of it, which we still call it $K_i$, with abuse of notation) such that $0 < \alpha\mu(K_i) \leq \mu(\Gamma \cap K_i) \leq \mu(J)$, with a suitable chosen $0 < \alpha < 1$. Choose a $\tau$ such that the set $K_i +\tau$ is contained in $J$. The set $\tilde{\Gamma}: = \Gamma \cap K_i \subset K_i$ is then translated in $J$ and $\mu(\tilde{\Gamma} + \tau) = \mu(\tilde{\Gamma}) > 0$. Setting $\phi : \tilde{\Gamma} \rightarrow \{\tilde{\Gamma} + \tau\}$, $x \mapsto x + \tau$, we define a new function $\tilde{g}$   
\begin{equation}\label{g tilde}
\tilde{g}(x) = \begin{cases} 0 & \mbox{if } x \in \tilde{\Gamma} \\
g(\phi^{-1}(x))I_{\tilde{\Gamma}}(\phi^{-1}(x)) & \mbox{if } x \in \tilde{\Gamma} +\tau \\
g(x) &\mbox{otherwise. } \end{cases} 
\end{equation}
It is easily seen that the constraints of the problem are satisfied and $\int_{D}g(x)b(x) d\mu < \int_{D}\tilde{g}(x)b(x) d\mu$, because $\forall x \in B$ and $\forall y \in D\setminus B$, $g(y)b(y) < g(y)b(x)$.

Now, let us suppose that the set $B \setminus (B \cap G)$ does not contain intervals. If there is a subset $F$ in $B \setminus (B \cap G)$ and an interval $E$, such that $\mu(E \triangle F) = 0$, we can adopt the same construction of the case above, with $J$ replaced by $E$. 

Finally, we discuss the last case. $B$ is open in $D$, therefore there exists a countable family of open intervals $\{O_n\}_{n \in \mathbb{N}}$ such that $B \subset \bigcup_{n \in \mathbb{N}} O_n$ and we further suppose that, for any $n$, the intersection $B \setminus (B \cap G) \cap O_n$ has never full measure. We can choose a suitable $K_i$ (as defined above), a $O_n$ (or suitable subintervals of those sets) and $0 < \alpha_1, \alpha_2 < 1$, for which:
\begin{equation}\label{eq: halmos inequality Gamma tilde}
0 < \alpha_1\mu(K_i) \leq \mu(\tilde{\Gamma}) \leq \mu(K_i),
\end{equation}
\begin{equation}\label{eq: halmos inequality O}
0 < \alpha_2\mu(O_n) \leq \mu(B \setminus (B \cap G) \cap O_n) \leq \mu(O_n),
\end{equation}
and with
\begin{equation}\label{eq: uguaglianze misure}
\mu(K_i)\leq \mu(O_n).
\end{equation}

Now, let us take a translation $\tau$ mapping $K_i$ one-to-one and into $O_n$ and such that $\{\tilde{\Gamma}+\tau\} \cap \left( B \setminus (B \cap G) \right)  \cap O_n \neq \emptyset$. If we can show that the set $\{\tilde{\Gamma}+\tau\} \cap \left( B \setminus (B \cap G) \right) \cap O_n$ has positive measure, then we can re-define the map $\phi : \tilde{\Gamma} \cap \left\lbrace \left(  B \setminus (B \cap G) \cap O_n \right)  -\tau \right\rbrace  \rightarrow \{\tilde{\Gamma} + \tau\} \cap \left( B \setminus (B \cap G)\right) \cap O_n$ and use the same construction of $\tilde{g}$ to arrive to the same conclusion. From (\ref{eq: halmos inequality O}) and the properties of measure, we have:
\begin{eqnarray}\label{eq: disuguaglianza halmos totale}
\mu(\{\tilde{\Gamma} + \tau\} \cap \left(  B \setminus (B \cap G) \right)  \cap O_n) & \geq & \alpha_2 \mu(\{\tilde{\Gamma} + \tau\}\cap O_n) \\ & = & \alpha_2 \mu(\tilde{\Gamma}\cap \left\lbrace O_n - \tau \right\rbrace ) \nonumber \\
& \geq & \alpha_2\mu(\tilde{\Gamma} \cap K_i) \nonumber \\
& \geq & \alpha_2\alpha_1 \mu(K_i) > 0\nonumber
\end{eqnarray}
and so the claim follows.

The previous discussion allows us to consider the case with $G$ and $B$ having empty intersection. We can use the same procedure used in the case above, with $B \setminus (B \cap G)$ containing an interval.  

Therefore, we can assume that, necessarily, we must have $\mu(G \triangle B) = 0$. Now, given that assumption, we show that if $g(x) \neq w(x)$ on a set of positive measure, $g(x)$ cannot be the maximum for the functional in (\ref{eq:infinitesparsepb_solow_new}).

Suppose that $\mu(G \triangle B) = 0$ and $\exists\ U \subset B$, such that $\mu(U) > 0$ with
\begin{equation}\label{def g (1)}
g(x) = \begin{cases} w(x) &\mbox{if } x \in D\setminus U \\ 
w(x) + h(x) & \mbox{if } x \in U, \end{cases} 
\end{equation}
with $h(x) \in L^{\infty}(U)$ and $h(x) > 0 $ a.e. Then, we would obtain:
\begin{eqnarray*}
\Vert g(x) \Vert_{L^2(G)} & = & \int_{B\setminus U} w(x)^2 d\mu + \int_{U} \left( w(x) + h(x)\right) ^2 d\mu \\
& \geq & 1 + \int_{U} h(x)^2 d\mu > 1.
\end{eqnarray*}
The case with 
\begin{equation}\label{def g trivial}
g(x) = \begin{cases} w(x) &\mbox{if } x \in D\setminus U \\ 
w(x) - h(x) & \mbox{if } x \in U, \end{cases} 
\end{equation}
$0 \leq h(x) \leq w(x)I_{U}(x)$, is clearly rejected since such a $g(x)$ would not be the solution for (\ref{eq:infinitesparsepb_solow_new}).

Suppose, now, that $\mu(G \triangle B) = 0$ and $\exists\ U, V \subset B$, such that $\mu(U) > 0,\ \mu(V) > 0$ with
\begin{equation}\label{def g (2)}
g(x) = \begin{cases} w(x) &\mbox{if } x \in D\setminus (U \cup V) \\ 
w(x) + f(x) & \mbox{if } x \in U \\
w(x) - h(x) & \mbox{if } x \in V \end{cases} 
\end{equation}
with $f(x), h(x) \in L^{\infty}(U)$ and $f(x), h(x) > 0 $ a.e. The two functions $f$ and $h$ have to satisfy some constraints in order for $g$ to be admissible. So, let us compute the $L^2$ norm of the defined $g$:
\begin{eqnarray*}
\Vert g(x) \Vert^2_{L^2(G)} & = & \int_{D\setminus (U\cup V)} w(x)^2 d\mu + \int_{U} \left( w(x) + f(x)\right) ^2 d\mu \\
& + & \int_{V} \left( w(x) - h(x)\right) ^2 d\mu.
\end{eqnarray*} 
Rearranging the integrals and substituting the definition of $w$, we obtain
\begin{eqnarray}\label{norm of g}
\Vert g(x) \Vert^2_{L^2(G)} & = & 1 + \Vert f(x) \Vert^2_{L^2(U)} + \Vert h(x) \Vert^2_{L^2(V)}\\ \nonumber & + & \frac{2}{\Vert b(x) \Vert_{L^2(B)}}\left( \int_{U}f(x)b(x) d\mu - \int_{V}h(x)b(x) d\mu \right). 
\end{eqnarray}
For $g$ to be admissible, then, we must have:
\begin{eqnarray}\label{vincolo}
\Vert f(x) \Vert^2_{L^2(U)} + \Vert h(x) \Vert^2_{L^2(V)} + \\ \nonumber  \frac{2}{\Vert b(x) \Vert_{L^2(B)}}\left( \int_{U}f(x)b(x) d\mu - \int_{V}h(x)b(x) d\mu \right) = 0. 
\end{eqnarray}  
which leads to:
\begin{eqnarray}\label{solution vincolo}
\left( \int_{U}f(x)b(x) d\mu - \int_{V}h(x)b(x) d\mu \right) = \\ \nonumber -\frac{\Vert b(x) \Vert_{L^2(B)}}{2}\left( \Vert f(x) \Vert^2_{L^2(U)} + \Vert h(x) \Vert^2_{L^2(V)}\right). 
\end{eqnarray}  
Note that the second member is negative. Now, if we evaluate the functional in (\ref{eq:infinitesparsepb_solow_new}) with such a $g$, we have
\begin{eqnarray}
\int_{D}g(x)b(x) d\mu & = & \int_{B}w(x)b(x) d\mu + \int_{U}f(x)b(x) d\mu - \int_{V}h(x)b(x) d\mu \nonumber \\
& = & \Vert b(x) \Vert_{L^2(B)} -\frac{\Vert b(x) \Vert_{L^2(B)}}{2}\left( \Vert f(x) \Vert^2_{L^2(U)} + \Vert h(x) \Vert^2_{L^2(V)}\right) \nonumber \\
& < & \Vert b(x) \Vert_{L^2(B)} = \int_{D}w(x)b(x) d\mu.
\end{eqnarray}
Thus $g$ cannot be the maximum. We therefore conclude that $w(x)$, as we defined it, has to be the solution to the problem.
\end{proof}

\bibliographystyle{chicago}
\bibliography{myrefs}

\begin{thebibliography}{}

\bibitem[\protect\citeauthoryear{Boyd and Vandenberghe}{Boyd and
  Vandenberghe}{2004}]{BoydVande}
Boyd, S. and L.~Vandenberghe (2004).
\newblock {\em Convex Optimization}.
\newblock Cambridge University Press.

\bibitem[\protect\citeauthoryear{Celeux, Martin-Magniette, Maugis, and
  Raftery}{Celeux et~al.}{2013}]{CMMRR2013}
Celeux, G., M.~Martin-Magniette, C.~Maugis, and A.~Raftery (2013).
\newblock Comparing model selection and regularization approaches to variable
  selection in model-based clustering.
\newblock {\em arXiv:1307.7860\/}.

\bibitem[\protect\citeauthoryear{Chang}{Chang}{1983}]{Chang83}
Chang, W. (1983).
\newblock On using principal components before separating a mixture of two
  multivariate normal distributions.
\newblock {\em Journal of the Royal Statistical Society, Ser. C\/}~{\em 32},
  267--275.

\bibitem[\protect\citeauthoryear{Chen, Reiss, and Tarpey}{Chen
  et~al.}{2014}]{Tarpey14}
Chen, H., P.~Reiss, and T.~Tarpey (2014).
\newblock Optimally weighted ${L}^2$ distance for functional data.
\newblock {\em Biometrics\/}~{\em 70}, 516--525.

\bibitem[\protect\citeauthoryear{Fraley and Raftery}{Fraley and
  Raftery}{2002}]{Fraley-Raftery-2002}
Fraley, C. and A.~E. Raftery (2002).
\newblock Model-based clustering, discriminant analysis, and density
  estimation.
\newblock {\em Journal of the American Statistical Association\/}~{\em
  97\/}(458), 611--631.

\bibitem[\protect\citeauthoryear{Friedman and Meulman}{Friedman and
  Meulman}{2004}]{fm}
Friedman, J.~H. and J.~J. Meulman (2004).
\newblock Clustering objects on a subset of attributes.
\newblock {\em Journal of the Royal Statistical Society, Ser. B\/}~{\em 66},
  815--849.

\bibitem[\protect\citeauthoryear{Gosh and Chinnaiyan}{Gosh and
  Chinnaiyan}{2002}]{GC2002}
Gosh, D. and A.~Chinnaiyan (2002).
\newblock Mixture modelling of gene expression data from microarray
  experiments.
\newblock {\em Bioinformatics\/}~{\em 18}, 275--286.

\bibitem[\protect\citeauthoryear{Hartigan}{Hartigan}{1975}]{har}
Hartigan, J. (1975).
\newblock {\em Clustering Algorithms}.
\newblock Wiley \& Sons Inc.

\bibitem[\protect\citeauthoryear{Hartigan}{Hartigan}{1978}]{Hartigan-1978}
Hartigan, J.~A. (1978).
\newblock Asymptotic distributions for clustering criteria.
\newblock {\em The Annals of Statistics\/}~{\em 6\/}(1), 117--131.

\bibitem[\protect\citeauthoryear{James and Sugar}{James and
  Sugar}{2003}]{James2013}
James, G. and C.~Sugar (2003).
\newblock Clustering for sparsely sampled functional data.
\newblock {\em Journal of the American Statistical Association\/}~{\em
  98\/}(462), 397--408.

\bibitem[\protect\citeauthoryear{Liu, Zhang, Palumbo, and Lawrence}{Liu
  et~al.}{2003}]{Liu2003}
Liu, J., J.~Zhang, M.~Palumbo, and C.~Lawrence (2003).
\newblock Bayesian clustering with variable and transformation selections.
\newblock In {\em Bayesian Statistics 7}, eds. J.M. Bernardo et al., pp.\
  249--275. Clarendon press.

\bibitem[\protect\citeauthoryear{Luss and d'Aspremont}{Luss and
  d'Aspremont}{2010}]{Luss.et.al.10}
Luss, R. and A.~d'Aspremont (2010).
\newblock Clustering and feature selection using sparse principal component
  analysis.
\newblock {\em Optimization and Engineering\/}~{\em 11\/}(1), 145--157.

\bibitem[\protect\citeauthoryear{Maugis, Celeux, and Martin-Magniette}{Maugis
  et~al.}{2009a}]{MCMM2009a}
Maugis, C., G.~Celeux, and M.~Martin-Magniette (2009a).
\newblock Variable selection for clustering with gaussian mixture models.
\newblock {\em Biometrics\/}~{\em 65}, 701--709.

\bibitem[\protect\citeauthoryear{Maugis, Celeux, and Martin-Magniette}{Maugis
  et~al.}{2009b}]{MCMM2009b}
Maugis, C., G.~Celeux, and M.~Martin-Magniette (2009b).
\newblock Variable selection in model-based clustering: a general variable role
  modeling.
\newblock {\em Computational Statistics and Data Analysis\/}~{\em 52},
  3872--3882.

\bibitem[\protect\citeauthoryear{McLachlan and Peel}{McLachlan and
  Peel}{2000}]{McLachlan-Peel-2000}
McLachlan, G. and D.~Peel (2000).
\newblock {\em Finite Mixture Models}.
\newblock Wiley Series in Probability and Statistics.

\bibitem[\protect\citeauthoryear{Pan and Shen}{Pan and Shen}{2007}]{PS2007}
Pan, W. and X.~Shen (2007).
\newblock Penalized model-based clustering with application to variable
  selection.
\newblock {\em Journal of Machine Learning Research\/}~{\em 8}, 1145--1164.

\bibitem[\protect\citeauthoryear{Pollard}{Pollard}{1981}]{po2}
Pollard, D. (1981).
\newblock Strong consistency of k-means clustering.
\newblock {\em The Annals of Statistics\/}~{\em 9\/}(1), 135--140.

\bibitem[\protect\citeauthoryear{{R Development Core Team}}{{R Development Core
  Team}}{2011}]{R}
{R Development Core Team} (2011).
\newblock {\em R: A Language and Environment for Statistical Computing}.
\newblock Vienna, Austria: R Foundation for Statistical Computing.
\newblock {ISBN} 3-900051-07-0.

\bibitem[\protect\citeauthoryear{Raftery and Dean}{Raftery and
  Dean}{2006}]{RD2006}
Raftery, A. and N.~Dean (2006).
\newblock Variable selection for model-based clustering.
\newblock {\em Journal of the American Statistical Association\/}~{\em 101},
  168--178.

\bibitem[\protect\citeauthoryear{Ramsay and Silverman}{Ramsay and
  Silverman}{2005}]{rs2}
Ramsay, J.~O. and B.~W. Silverman (2005).
\newblock {\em Functional Data Analysis}.
\newblock Springer.

\bibitem[\protect\citeauthoryear{Ramsay and Wickham}{Ramsay and
  Wickham}{2007}]{fda-package}
Ramsay, J.~O. and H.~Wickham (2007).
\newblock {\em fda. Functional Data Analysis}.
\newblock R package version 3.0.2.

\bibitem[\protect\citeauthoryear{Rand}{Rand}{1971}]{Rand71}
Rand, W. (1971).
\newblock Objective criteria for the evaluation of clustering methods.
\newblock {\em Journal of the American Statistical Association\/}~{\em 66},
  846--850.

\bibitem[\protect\citeauthoryear{Sangalli, Secchi, Vantini, and
  Vitelli}{Sangalli et~al.}{2010}]{ssvv}
Sangalli, L.~M., P.~Secchi, S.~Vantini, and V.~Vitelli (2010).
\newblock K-mean alignment for curve clustering.
\newblock {\em Computational Statistics and Data Analysis\/}~{\em 54},
  1219--1233.

\bibitem[\protect\citeauthoryear{Tarpey and Kinateder}{Tarpey and
  Kinateder}{2003}]{tk}
Tarpey, T. and K.~K.~J. Kinateder (2003).
\newblock Clustering functional data.
\newblock {\em Journal of Classification\/}~{\em 20}, 93--114.

\bibitem[\protect\citeauthoryear{Tibshirani, Walther, and Hastie}{Tibshirani
  et~al.}{2001}]{twh2001}
Tibshirani, R., G.~Walther, and T.~Hastie (2001).
\newblock Estimating the number of clusters in a dataset via the gap statistic.
\newblock {\em Journal of the Royal Statistical Society, Ser. B\/}~{\em
  32\/}(2), 411--423.

\bibitem[\protect\citeauthoryear{Tuddenham and Snyder}{Tuddenham and
  Snyder}{1954}]{Tuddenham-Snyder-1954}
Tuddenham, R.~D. and M.~M. Snyder (1954).
\newblock Physical growth of california boys and girls from birth to age 18.
\newblock Technical Report~1, University of California Publications in Child
  Development.

\bibitem[\protect\citeauthoryear{Wang and Zhu}{Wang and Zhu}{2008}]{WZ2008}
Wang, S. and J.~Zhu (2008).
\newblock Variable selection for model-based high dimensional clustering and
  its application to microarray data.
\newblock {\em Biometrics\/}~{\em 64}, 440--448.

\bibitem[\protect\citeauthoryear{Witten and Tibshirani}{Witten and
  Tibshirani}{2010}]{twitt}
Witten, D. and R.~Tibshirani (2010).
\newblock A framework for feature selection in clustering.
\newblock {\em Journal of American Statistical Association\/}~{\em 105\/}(490),
  713--726.

\bibitem[\protect\citeauthoryear{Xie, Pan, and Shen}{Xie
  et~al.}{2008}]{XPS2008}
Xie, B., W.~Pan, and X.~Shen (2008).
\newblock Penalized model-based clustering with cluster-specific diagonal
  covariance matrices and grouped variables.
\newblock {\em Electronic Journal of Statistics\/}~{\em 2}, 168--212.

\end{thebibliography}

\end{document}